\documentclass[11pt]{article}

\usepackage[margin=1.25in]{geometry}
\usepackage{verbatim}
\usepackage{amsmath}
\usepackage{amsthm}
\usepackage{amssymb}
\usepackage{graphicx}
\usepackage[table]{xcolor}
\usepackage{multirow} 
\usepackage{graphicx}
\usepackage{times} 
\usepackage{mathtools}
\usepackage{url}
\usepackage{amstext}
\usepackage{esint}
\usepackage[round]{natbib}

\usepackage{rotating}

\global\long\def\R{\mathbb{R}}

\usepackage[font=scriptsize]{caption}

\makeatletter

\providecommand{\tabularnewline}{\\}

\newtheorem{thm}{Theorem}
\theoremstyle{definition}
 \newtheorem{example}{\protect\examplename}
\theoremstyle{definition}
\newtheorem{defn}{\protect\definitionname}
\theoremstyle{remark}
\newtheorem{rem}{\protect\remarkname}
\theoremstyle{remark}

\theoremstyle{remark}

\usepackage{hyperref}

\DeclareMathOperator{\argmax}{argmax}

\makeatother

\providecommand{\claimname}{Claim}
\providecommand{\definitionname}{Definition}
\providecommand{\examplename}{Example}
\providecommand{\remarkname}{Remark}

\begin{document}
\global\long\def\alp{\alpha}%
\global\long\def\gam{\gamma}%
\global\long\def\Gam{\Gamma}%
\global\long\def\del{\delta}%
\global\long\def\Del{\Delta}%
\global\long\def\eps{\varepsilon}%
\global\long\def\lam{\lambda}%
\global\long\def\kap{\kappa}%
\global\long\def\sig{\sigma}%
\global\long\def\Sig{\Sigma}%
\global\long\def\vrp{\varphi}%
\global\long\def\tht{\theta}%
\global\long\def\omg{\omega}%
\global\long\def\Omg{\Omega}%

\global\long\def\til#1{\tilde{#1}}%
 
\global\long\def\wtil#1{\widetilde{#1}}%
 
\global\long\def\what#1{\widehat{#1}}%

\global\long\def\C{\mathbb{C}}%
 
\global\long\def\N{\mathbb{N}}%
 
\global\long\def\Q{\mathbb{Q}}%
 
\global\long\def\R{\mathbb{R}}%
 
\global\long\def\Z{\mathbb{Z}}%
\global\long\def\HC{\mathbb{\mathcal{H}}}%
\global\long\def\M{\mathbb{\mathcal{M}}}%
\global\long\def\W{\mathbb{\mathcal{W}}}%

\global\long\def\bb{\boldsymbol{b}}%
\global\long\def\fb{\boldsymbol{f}}%
\global\long\def\ii{\boldsymbol{i}}%
\global\long\def\mm{\boldsymbol{m}}%
\global\long\def\nn{\boldsymbol{n}}%
\global\long\def\NN{\boldsymbol{N}}%
\global\long\def\pp{\boldsymbol{p}}%
\global\long\def\ss{\boldsymbol{s}}%
\global\long\def\tb{\boldsymbol{t}}%
\global\long\def\uu{\boldsymbol{u}}%
\global\long\def\vv{\boldsymbol{v}}%
\global\long\def\VV{\boldsymbol{V}}%
\global\long\def\ww{\boldsymbol{w}}%
\global\long\def\WW{\boldsymbol{W}}%
\global\long\def\xx{\boldsymbol{x}}%
\global\long\def\XX{\boldsymbol{X}}%
\global\long\def\yy{\boldsymbol{y}}%
\global\long\def\YY{\boldsymbol{Y}}%
\global\long\def\zz{\boldsymbol{z}}%
\global\long\def\ZZ{\boldsymbol{Z}}%
\global\long\def\alpp{\boldsymbol{\alpha}}%
\global\long\def\betaa{\boldsymbol{\beta}}%
\global\long\def\epss{\boldsymbol{\eps}}%
\global\long\def\pii{\boldsymbol{\pi}}%

\global\long\def\ra{\rightarrow}%
 
\global\long\def\la{\leftarrow}%
 
\global\long\def\lra{\longrightarrow}%
 
\global\long\def\Lra{\Longrightarrow}%
\global\long\def\rla{\leftrightarrow}%
\global\long\def\ds{\text{{diag(\ensuremath{\ss})}}}%
\global\long\def\BC{Barber and Cand\'es}%
\global\long\def\sign{\text{{sign}}}%
\global\long\def\diag{\text{{diag}}}%
\global\long\def\supp{\text{{support}}}%

\global\long\def\supsec{Supplementary Section}%
\global\long\def\supfig{Supplementary Figure}%

\title{Multiple competition-based FDR control for peptide detection}

\author{Kristen Emery$^1$, Syamand Hasam$^1$, William Stafford Noble$^2$, Uri Keich$^1$\\
$^1$School of Mathematics and Statistics F07\\
University of Sydney\\
$^2$Departments of Genome Sciences and of Computer Science and Engineering\\
University of Washington\\
}

\maketitle

\begin{abstract}
	Competition-based FDR control has been commonly used for over a decade
	in the computational mass spectrometry community~\citep{elias:target}.
	Recently, the approach has gained significant popularity in other fields after \cite{barber:controlling} 
	laid its theoretical foundation in a more general setting that included the feature selection problem.
	In both cases, the competition is based on a head-to-head comparison between an observed
	score and a corresponding decoy / knockoff.
	\cite{keich:progressive} recently demonstrated some advantages of using multiple rather than
	a single decoy when addressing the problem of assigning peptide sequences to observed mass spectra.
	In this work, we consider a related problem --- detecting peptides based on a collection of mass spectra --- and we
	develop a new framework for competition-based FDR control using multiple null scores. Within this framework, we offer several
	methods, all of which are based on a novel procedure that rigorously controls the
	FDR in the finite sample setting. Using real data to study the peptide detection problem we show that,
	relative to existing single-decoy methods, our approach can increase the number of discovered peptides
	by up to 50\% at small FDR thresholds.
\end{abstract}

\noindent
{\sc Keywords:} multiple hypothesis testing, peptide detection, tandem mass spectrometry, false discovery rate

\section{Introduction}

Proteins are the primary functional molecules in living cells, and
tandem mass spectrometry (MS/MS) currently provides the most efficient means of studying proteins in a high-throughput fashion. 
Knowledge of the protein complement in a cellular population provides insight into the functional state of the cells.
Thus, MS/MS can be used to functionally characterize cell types, differentiation stages, disease states, or species-specific differences.
For this reason, MS/MS is the driving technology for much of the rapidly growing field of proteomics.

Paradoxically, MS/MS does not measure proteins directly.
Because proteins themselves are difficult to separate and manipulate biochemically, an MS/MS experiment involves first digesting proteins into smaller pieces, called ``peptides.''  The peptides are then measured directly.  A typical MS/MS experiment generates $\sim$10 observations (``spectra'') per second, so a single 30-minute MS/MS experiment will generate approximately 18,000 spectra.
Canonically, each observed spectrum is generated by a single peptide.
Thus, the first goal of the downstream analysis is to infer which peptide was responsible for generating each observed spectrum.
The resulting set of detected peptides can then be used, in a second analysis stage, to infer what proteins are present in the sample.

In this work, we are interested in the first problem --- peptide detection. 
Specifically, we focus on the task of assigning confidence estimates to peptides that have been identified by MS/MS.
As is common in many molecular biology contexts, these confidence estimates are typically reported in terms of the false discovery rate (FDR),
i.e., the expected value of the proportion of false discoveries among a set of detected peptides.
For reasons that will be explained below, rather than relying on standard methods for control of the FDR such as the \cite{benjamini:controlling} (BH) procedure the proteomics field employs a strategy known as ``target-decoy competition'' (TDC) to control the FDR in
the reported list of detected peptides~\citep{elias:target}.
TDC works by comparing the list of peptides detected with a list of artificial peptides, called ``decoys,'' detected using the same spectra set.
The decoys are created by reversing or randomly shuffling the letters of the real (``target'') peptides.
The TDC protocol, which is described in detail in Section~\ref{sec:TDC}, estimates the FDR by counting the number of detected decoy peptides and using this count as an estimate for the number of incorrectly detected target peptides.

One clear deficiency of TDC is its reliance on a single set of decoy peptides to estimate the FDR.
Thus, with ever increasing computational resources one can ask whether we can gain something by exploiting multiple randomly drawn decoys for each target peptide.
\cite{keich:progressive} and \cite*{keich:averaging} recently described such a procedure, called ``average target-decoy competition'' (aTDC), that, in the context of the related spectrum identification problem (described in Section~\ref{sec:shotgun}), reduces the variability associated with TDC and can provide a modest boost in power.

In this paper we propose a new approach to using multiple decoy scores.
The proposed procedure relies on a direct competition between
the target and its corresponding decoy scores, rather than on averaging single competitions. We formulate our approach
in the following more general setting.
Suppose that we can compute a test statistic $Z_{i}$ for each null hypothesis $H_{i}$,
so that the larger $Z_{i}$ is, the less likely is the null. However,
departing from the standard multiple hypotheses setup, we further assume that we cannot
compute p-values for the observed scores. Instead, we can only generate
a \emph{small} sample of independent decoys or competing null scores for each
hypothesis $H_{i}$: $\til Z_{i}^{j}$ $j=1,\dots,d$ (Definition \ref{iidDecoys}).
Note that the case $d=1$ corresponds to the TDC setup described above.
We will show using both simulated and real data that the novel method we propose yields more power (more discoveries)
than our aforementioned averaging procedure.

In addition to the peptide detection problem, our proposed procedure is applicable in several other bioinformatics applications.  For example, the procedure could be used when analyzing a large number
of motifs reported by a motif finder, e.g., \cite{harbison:transcriptional},
where creating competing null scores can require the time consuming
task of running the finder on randomized versions of the input sets,
e.g., \cite{ng:gimsan}.  In addition, our procedure is applicable to controlling the FDR in selecting
differentially expressed genes in microarray experiments where a small
number of permutations is used to generate competing null scores~\citep*{Tusher:Significance}.

Our proposed method can also be viewed as a generalization of the ``knockoff'' procedure of \cite{barber:controlling}.
The knockoff procedure is a competition-based FDR control method that was initially developed for 
feature selection in a classical linear regression model.
The procedure has gained a lot of interest
in the statistical and machine learning communities, where it has 
has been applied to various applications in biomedical research~\citep{xiao:mapping, gao:model, read:predicting} and has been extended to work in conjunction with deep neural networks~\citep{lu:deeppink} and with time series data~\citep{fan:ipad}.
Despite the different terminology, both knockoffs and decoys serve the same purpose in competition-based FDR control; hence, for the ideas presented in this paper the two are interchangeable. 

A significant part of \BC' work is the sophisticated construction of their knockoff scores;
controlling the FDR then follows exactly the same competition
that TDC uses.  Indeed, their Selective SeqStep+ (SSS+) procedure rigorously formalizes in a much more general setting
the same procedure described above in the context of TDC.

Note that \BC\ suggested that using multiple knockoffs could improve the power of their procedure so the methods
we propose here could provide a stepping stone toward that. However, we would still need to figure out
how to generalize their construction from one to multiple knockoffs.

\section{Background}

By way of background, we begin with a brief description of the mass spectrometry protocol and outline how the resulting spectra are assigned to peptides. We then discuss existing methods for statistical confidence estimation, and we introduce the peptide detection problem which is the focus of our work.

\subsection{Shotgun proteomics and spectrum identification}\label{sec:shotgun}

In a ``shotgun proteomics'' MS/MS experiment, proteins in a complex biological sample are extracted and digested into peptides, each with an associated charge.
These charged peptides, called ``precursors,'' are measured by the mass spectrometer, and a subset of the precursors
are then selected for further fragmentation into charged ions, which are detected and recorded by a second round of mass spectrometry
\citep{hernandez:automated,noble:computational}.  The recorded tandem fragmentation spectra, or spectra for short, are
then subjected to computational analysis.

This analysis typically begins with the spectrum identification problem, which involves inferring which peptide was responsible for
generating each observed fragmentation spectrum. 
The most common solution to this problem is peptide database search.
Pioneered by SEQUEST~\citep{eng:approach}, the search engine extracts from the peptide database all ``candidate peptides,''
defined by having their mass lie within a pre-specified tolerance of the measured precursor mass. The quality of the match between each one of these candidate peptides and the observed fragmentation
spectrum is then evaluated using a score function. Finally, the optimal peptide-spectrum match (PSM) for the given spectrum
is reported, along with its score~\citep{nesvizhskii:survey}.

In practice, many expected fragment ions will fail to be observed for any given spectrum, and the spectrum is also likely to contain a
variety of additional, unexplained peaks~\citep{noble:computational}.
Hence, sometimes the reported PSM is correct --- the peptide assigned to the spectrum was
present in the mass spectrometer when the spectrum was generated --- and sometimes the PSM is incorrect.
Ideally, we would report only the correct PSMs, but obviously we are not privy to this information: all we have is the score of the PSM,
indicating its quality.
Therefore, we report a thresholded list of top-scoring PSMs, together with the critical estimate of the fraction of incorrect
PSMs in our reported list.

\subsection{False discovery rate control in spectrum identification}\label{sec:TDC}

The general problem of controlling the proportion of false discoveries has been studied extensively
in the context of multiple hypotheses testing (MHT). Briefly reviewed in Section~\ref{sec:MHT} below,
this setup does not apply directly to the spectrum identification problem.
A major reason for that is
the presence in any shotgun proteomics dataset of both ``native spectra''
(those for which their generating peptide is in the target database) and ``foreign spectra''
(those for which it is not). These create different types of false positives, implying
that we typically cannot apply FDR controlling procedures that were designed
for the general MHT context to the spectrum identification problem~\citep{keich:controlling}.

Instead, the mass spectrometry community uses TDC to control the FDR in the
reported list of PSMs~\citep{elias:target,cerqueira:mude,jeong:false,elias:target2}.
TDC works by comparing searches against a target peptide database with searches against a decoy
database of peptides obtained from the original database by randomly
shuffling (or reversing) each peptide in the target database.

More precisely, let $Z_i$ be the score of the optimal match (PSM) to the $i$th spectrum in the target database,
and let $\til Z_{i}$ be the corresponding optimal match in the decoy database.
Each decoy score $\til Z_{i}$ directly competes with its corresponding
target score $Z_{i}$ for determining the reported list of discoveries.
Specifically, for each score threshold $T$ we only report target
PSMs that won their competition: $Z_{i}>\max\{T,\til Z_{i}\}$.
Subsequently, the number of decoy wins ($\til Z_{i}>\max\{T,Z_{i}\}$) is used to
estimate the number of false discoveries in the list of target wins.
Thus, the ratio between that estimate and the number of target wins
yields an estimate of the FDR among the target wins. To control the
FDR at level $\alp$ we choose the smallest threshold $T=T(\alpha)$
for which the estimated FDR is still $\le\alpha$.

It was recently
shown that, assuming that incorrect PSMs are independently equally likely to come
from a target or a decoy match and provided we add 1 to
the number of decoy wins before dividing by the number of target wins, this procedure
rigorously controls the FDR~\citep{he:theoretical,levitsky:unbiased}.

\subsection{The peptide detection problem}

The spectrum identification is largely used as the first step in addressing 
the peptide identification problem that motivates the research presented here. 
Indeed, to identify the peptides we begin, just like we do in spectrum identification,
by assigning each spectrum to the unique target/decoy peptide
which offers the best match to this spectrum in the corresponding database. We then assign to each target peptide
a score $Z_j$ which is the maximum of all PSM scores of spectra that were assigned to this peptide in the first phase.
Similarly, we assign to the corresponding decoy peptide a score $\til Z_j$, which again is the maximum of all
PSM scores involving spectra that were assigned to that decoy peptide. The rest continues using the same TDC protocol we
outlined above for the spectrum identification problem~\citep{granholm:determining,savitski:scalable}.

\section{Controlling the FDR using multiple decoys\label{sec:Multiple-decoys-FDR}}


\subsection{Why do we need a new approach?}
\label{sec:MHT}

In the multiple hypotheses testing (MHT) set up we simultaneously test $m$ (null)
hypotheses $H_{1},\dots,H_{m}$, looking to reject as many as possible
subject to some statistical control of our error rate.
Pioneered by Benjamini and Hochberg, the common approach to controlling
the error rate in the MHT context is through bounding the expected proportion of false
discoveries at any desired level $\alpha\in(0,1)$.

More precisely,
assume we have a selection procedure that produces a list of $R$
discoveries of which, unknown to us, $V$ are false, and let $Q=V/\max\left\{ R,1\right\}$
be the unobserved false discovery proportion (FDP).
\cite{benjamini:controlling} showed that applying their selection procedure (BH) at level
$\alpha$ controls $E(Q)$, the false discovery rate (FDR): $E(Q)\le\alpha$.

Other, more powerful selection procedures that rely on estimating
$\pi_{0}$, the fraction of true null hypotheses, are available. Generally
referred to as ``adaptive BH'' procedures, with one particularly
popular variant by Storey,
these procedures are also predicated on our ability to assign a p-value
to each of our tested hypotheses \citep*[e.g.,][]{benjamini:adaptive,benjamini:adaptiveLinear,storey:direct,storey:strong}.
Hence, in particular they cannot be directly applied in our competition-based setup.

A key feature of our problem is that due to computational costs the number of decoys, $d$, is small.
Indeed, if we are able to generate a large number of independent decoys for each hypothesis, then we
can simply apply the above standard FDR controlling procedures to
the empirical p-values. These p-values are estimated from the
empirical null distributions, which are constructed for each hypothesis
$H_{i}$ using its corresponding decoys.
Specifically, these empirical p-values take values of the
form $(d_1-r_i+1)/d_{1}$, where $d_{1}=d+1$, and $r_i\in\left\{ 1,\dots,d_{1}\right\}$
is the rank of the originally observed score
(``original score'' for short) $Z_{i}$ in the combined
list of $d_{1}$ scores: $\left(\til Z_{i}^{0}=Z_{i},\til Z_{i}^{1},\dots,\til Z_{i}^{d}\right)$ ($r_i=1$ is the lowest rank).
Using these p-values the BH procedure rigorously controls
the FDR, and Storey's method will asymptotically control the FDR as the number of hypotheses
$m\ra\infty$.

Unfortunately, because $d$ is small, applying those standard FDR control procedures to
the rather coarse empirical p-values may yield very low power. For example,
if $d=1$, each empirical p-value is either 1/2 or 1, and therefore
for many practical examples both methods will not be able to make
any discoveries at usable FDR thresholds.

Alternatively, one might consider pooling all the decoys regardless
of which hypothesis generated them. The pooled empirical p-values attain
values of the form $i/\left(m\cdot d+1\right)$ for $i=1,\dots,md+1$; hence,
particularly when $m$ is large, the p-values generally no longer
suffer from being too coarse. However, other significant problems arise when pooling the decoys.
These issues --- discussed in \supsec~\ref{sec:Failures-of-established} ---
imply that in general, applying BH or Storey's procedure to p-values that are
estimated by pooling the competing null scores can be problematic
both in terms of power and control of the FDR.

\subsection{A novel meta-procedure for FDR control using multiple decoys}
\label{sec:metaProc}

The main technical contribution of this paper is the introduction of several procedures
that effectively control the FDR in our multiple competition-based setup and that rely on the following meta-procedure.\\

\noindent
\emph{Input:} an original/target score $Z_i$ and $d$ competing null scores $\til Z_{i}^{j}$ for each null hypothesis $H_i$.\\

\noindent
\emph{Parameters:} an FDR threshold $\alp\in(0,1)$, two tuning parameters $c=i_c/d_1$ ($d_1=d+1$), the ``original/target win'' threshold,
and $\lam=i_\lam/d_1$, the ``decoy win'' threshold where $i_\lam,i_c\in\{1,\dots,d\}$ and $c \le \lam$, as well as
a (possibly randomized) mapping function $\vrp:\{1,\dots,d_{1}-i_{\lam}\} \mapsto\{d_{1}-i_{c}+1,\dots,d_{1}\}$.\\

\noindent
\emph{Procedure:}
\begin{enumerate}	
\item Each hypothesis $H_i$ is assigned an original/decoy win label:
\begin{equation}
L_{i}=\begin{cases}
1 & r_{i}\ge d_{1}-i_{c}+1 \qquad\text{(original win)}\\
0 & r_{i}\in\left(d_{1}-i_{\lam},d_{1}-i_{c}+1\right) \qquad\text{(ignored hypothesis)}\\
-1 & r_{i}\le d_{1}-i_{\lam} \qquad\text{(decoy win)}
\end{cases},\label{eq:def_Li_general_c_lam}
\end{equation}
where $r_{i}\in\{1,\dots,d_1\}$ is the rank of the original score when added to the list of its $d$ decoy scores.
\item Each hypothesis $H_i$ is assigned a score $W_{i}=\til Z_{i}^{\left(s_{i}\right)}$, where $\til Z_i^{(j)}$ is the $j$th order statistic or the $j$th
largest score among $\left(\til Z_{i}^{0}=Z_{i},\til Z_{i}^{1},\dots,\til Z_{i}^{d}\right)$, and the ``selected rank'', $s_{i}$, is defined as
\begin{equation}\label{def_s_i}
s_i=\begin{cases}
r_i & L_i=1 \text{ (so $W_i=Z_i$ in an original win)} \\
u_i & L_i = 0 \text{ (where $u_i$ is randomly chosen uniformly in $\left\{ d_{1}-i_{c}+1,\dots,d_{1}\right\}$)}\\
\vrp(r_i) & L_i=-1 \text{ (so $W_i$ coincides with a decoy score in a decoy win)} 
\end{cases}
\end{equation}

\item The hypotheses are reordered so that $W_{i}$ are decreasing, and the list of discoveries is defined as
the subset of original wins $D(\alp,c,\lam)\coloneqq\left\{ i\,:\,i\le i_{\alp c\lam},L_{i}=1\right\}$, where
\begin{equation}
i_{\alp c\lam}\coloneqq\max\left\{ i\,:\,\frac{1+\#\left\{ j\le i\,:\,L_{j}=-1\right\} }{\#\left\{ j\le i\,:\,L_{j}=1\right\} \vee1}\cdot\frac{c}{1-\lam}\le\alp\right\} .\label{eq:reject_criterion-general-c-lam}
\end{equation}
\end{enumerate}
We assume above that all ties in determining the ranks $r_i$, as well as the order of $W_i$, are broken randomly, although other ways to handle ties are possible (e.g., Section 8.3 in our technical report \citep{emery:multiple}).

Note that the hypotheses for which $L_i=0$ can effectively be ignored as
they cannot be considered discoveries nor do they factor in the numerator of \eqref{eq:reject_criterion-general-c-lam}.

Our procedures vary in how they define the (generally randomized) mapping function $\vrp$ (and hence $s_i$ in \eqref{def_s_i}), as
well as in how they
set the tuning parameters $c,\lam$. For example, in the case $d=1$ setting $c=\lam=1/2$ and $\vrp(1)\coloneqq2$ our meta-procedure
coincides with TDC.
For $d>1$ we have increasing flexibility with $d$, but one obvious generalization of TDC is to set $c=\lam=1/d_1$. In this case
$L_i=1$ if the original score is larger than all its competing decoys and otherwise $L_i=-1$. Thus, by definition,
$\vrp$ is constrained to the constant value $d_1$ so $s_i\equiv d_1$ and $W_i$ is always set to 
$Z_{i}^{(d_1)}=\max\left\{ \til Z_{i}^{0},\dots,\til Z_{i}^{d}\right\}$. Hence we refer to this as the ``max method.''
As we will see, the max method controls the FDR, but this does not hold for any choice of $c,\lam$ and $\vrp$.
The following section specifies a sufficient condition on $c,\lam$ and $\vrp$ that guarantees FDR control.

\subsection{Null labels conditional probabilities property}
\begin{defn}
	Let $N$ be the indices of all true null hypotheses. We say the null labels conditional probabilities property (NLCP) is
	satisfied if
	conditional on all the scores $\W=(W_1,\dots W_m)$ the random labels $\{L_i\,:\,i\in N\}$ are (i)
	independent and identically distributed (iid) with $P(L_i=1\mid\W) = c$ and $P(L_i=-1\mid\W) = 1-\lam$, and (ii)
	independent of the false null labels $\{L_i\,:\,i\notin N\}$.
\end{defn}
Note that in claiming that TDC controls the FDR we implicitly assume that a false match is equally likely to arise from
a target win as it is from a decoy win independently of all other scores \citep{he:theoretical}. This property coincides with the NLCP
with $d=1$ and $c=\lam=1/2$.
Our next theorem shows that the NLCP generally guarantees the FDR control of our meta-procedure. Specifically, we argue that
with NLCP established step 3 of our meta-procedure can be viewed as a special case of \cite{barber:controlling}'s SSS+
procedure  and its extension by \cite{lei:power}'s Adaptive SeqStep (AS).
Both procedures are designed for sequential hypothesis testing  where the order of the hypotheses is pre-determined -- by the scores $W_i$ in our case.
\begin{thm}
	\label{metaProcThm}
If the NLCP holds then our meta-procedure controls the FDR in a finite-sample setting, that is,
$E\left(|D(\alp,c,\lam)\cap N| / |D(\alp,c,\lam)|\right) \le \alp$, where the expectation
is taken with respect to \emph{all} the decoy draws.
\end{thm}
\noindent
Why does Theorem \ref{metaProcThm} make sense?
If the NLCP holds then a true null $H_i$ is an original win ($L_i=1$) with probability $c$ and is a decoy win with probability $1-\lam$. Hence, 
the factor $\frac{c}{1-\lam}$ that appears in \eqref{eq:reject_criterion-general-c-lam}
adjusts the observed number of decoy wins, $\#\left\{ j\le i\,:\,L_{j}=-1\right\}$,
to estimate the number of (unobserved) false original wins (those for which the corresponding $H_i$ is a true null).
Ignoring the +1 correction, the adjusted ratio of \eqref{eq:reject_criterion-general-c-lam} therefore estimates the FDR
in the list of the first $i$ original wins. The procedure simply takes the largest such list for which the estimated FDR
is $\le\alp$.
\begin{proof}
	To see the connection with SSS+ and AS we assign each hypothesis $H_i$ a p-value $p_i\coloneqq P\left(L_{i}\ge l \right)$. Clearly,
	if the NLCP holds then
	\begin{equation}
		p_{i} = \begin{cases}
		c & l=1\\
		\lam & l=0\\
		1 & l=-1
	\end{cases}.\label{eq:pi_general_c-1}
	\end{equation}
	Moreover, the NLCP further implies that for any $u\in\left(0,1\right)$ and $i\in N$, $P\left(p_{i}\le u \mid \W \right)\le u$,
	and that the true null labels $L_i$, and hence the true null p-values, $p_i$, are independent conditionally on $\W$.

	It follows that, even after sorting the hypotheses by the decreasing order of the scores $W_i$, the p-values
	of the true null hypotheses are still iid valid p-values that are independent from the false nulls.
	Hence our result follows from Theorem 3 (SSS+) of \cite{barber:controlling} for $c=\lam$, and more generally for $c\le\lam$
	from Theorem 1 (AS) of \cite{lei:power} (with $s=c$). 
\end{proof}
\begin{rem}
	With the risk of stating the obvious we note that one cannot simply apply SSS+ or AS
	by selecting $W_{i}=Z_{i}$ for all $i$ with the corresponding empirical p-values $(d_1-r_i+1)/d_1$. Indeed, in this
	case the order of the hypotheses (by $W_{i}$) is not independent of the true null p-values.
\end{rem}

\subsection{When does the NLCP hold for our meta-procedure?}

To further analyze the NLCP we make the following assumption on our decoys.
\begin{defn}[formalizing the multiple-decoy problem]
	\label{iidDecoys}
If the $d_1$ (original and decoy) scores corresponding to each true null hypothesis are iid independently of all other
scores then we say we have ``iid decoys''.
\end{defn}

It is clear that if we have iid decoys then for each fixed $i\in N$ the rank $r_i$ is uniformly 
distributed on $1,\dots,d_1$, and hence $P(L_i=1) = c$ and $P(L_i=-1) = 1-\lam$.
However, to determine whether or not $r_i$ is still uniformly distributed when conditioning on $\W$ we need to look at
the mapping function $\vrp$ as well.

More specifically, in the iid decoys case the conditional distribution of $\{L_i\,:\,i\in N\}$ given $\W$
clearly factors into the product of the conditional distribution of each true null $L_i$ given $W_i$: a true null's $L_i$ is
independent of all $\{L_j,W_j\,:\,j\ne i\}$. Thus, it suffices to show that $L_i$ is independent of $W_i$ for each $i\in N$.
Moreover, because $W_i$ is determined in terms of $s_i$ and the \emph{set} of scores $\left\{ \til Z_{i}^{0},\dots,\til Z_{i}^{d}\right\}$,
and because a true null's label $L_i$ and $s_i$ are independent of the last set (a set is unordered), it suffices to show that $L_i$ is independent of $s_i$.
Of course, $s_i$ is determined by $\vrp$ as specified in \eqref{def_s_i}.

For example, consider the max method where $s_i\equiv d_1$ (equivalently $\vrp\equiv d_1$):
in this case, $L_i$ is trivially independent of $s_i$ and hence by the above discussion the method controls the FDR.
In contrast, assuming $d_1$ is even and choosing $\vrp\equiv d_1$ with $c=\lam=1/2$ we see that the scores $\{W_i\,:\,i\in N, L_i=-1\}$
will generally be larger than the corresponding $\{W_i\,:\,i\in N, L_i=1\}$. Indeed, when $L_i=-1$ we always choose the maximal score
$W_i=Z_{i}^{(d_1)}$, whereas $W_i$ is one of the top half scores when $L_i=1$. Hence, $P(L_i=-1 \mid \text{ higher }W_i)>1/2$.

So how can we guarantee that the NLCP holds for pre-determined values of $c=i_c/d_1$ and $\lam=i_\lam/d_1$?
The next theorem provides a sufficient condition on $\vrp$ (equivalently on $s_i$) to ensure the property holds.
\begin{thm}
	\label{NLCPthm}
If the iid decoys assumption holds, and if for any $i\in N$ and $j\in\left\{ d_{1}-i_{c}+1,\dots,d_{1}\right\}$
\begin{equation}
P\left(s_{i}=j,r_{i}\le d_{1}-i_{\lam}\right)=P\left(s_{i}=j,L_{i}=-1\right)=\frac{d_{1}-i_{\lam}}{d_{1}\cdot i_{c}},\label{eq:decoy_win_pos_c_lam1}
\end{equation}
then the NLCP holds and hence our meta-procedure with those values of $c,\lam$ and $\vrp$ controls the FDR.
\end{thm}
\begin{proof}
By (\ref{eq:decoy_win_pos_c_lam1}), for any $i\in N$ and $j\in\left\{ d_{1}-i_{c}+1,\dots,d_{1}\right\}$,
\begin{align*}
P\left(L_{i}=1\left|\right.s_{i}=j\right) &= 
	\frac{P\left(s_{i}=j,L_{i}=1\right)}{\sum_{l\in\left\{ -1,0,1\right\}}P\left(s_{i}=j,L_{i}=l\right)} \\
	&=\frac{1/d_{1}}{(d_{1}-i_{\lam})/(d_{1}\cdot i_{c})+\left(i_{\lam}-i_{c}\right)/d_{1}\cdot1/i_{c}+1/d_{1}}=\frac{i_{c}}{d_{1}} = c,\\
P\left(L_{i}=-1\left|\right.s_{i}=j\right) &= \frac{\left(d_1-i_{\lam}\right)/\left(d_{1}\cdot i_{c}\right)}{(d_{1}-i_{\lam})/(d_{i}\cdot i_{c})+\left(i_{\lam}-i_{c}\right)/d_{1}\cdot1/i_{c}+1/d_{1}}=\frac{d_1-i_{\lam}}{d_{1}} = 1 - \lam.
\end{align*}
At the same time $P\left(L_i=1\left|\right.s_i=j\right) = 1$ for $j\in\left\{1,\dots,i_c\right\}$ always holds;
therefore, $L_i$ is independent of $s_i$ and by the above discussion the NLCP holds. Theorem \ref{metaProcThm} completes the proof.
\end{proof}

For any fixed values of $c,\lam$ we can readily define a randomized $\vrp=\vrp_u$ so that the NLCP holds: randomly and uniformly map
$\{1,\dots,d_{1}-i_{\lam}\}$ onto $\{d_{1}-i_{c}+1,\dots,d_{1}\}$. Indeed, in this case \eqref{eq:decoy_win_pos_c_lam1} holds:
\begin{equation}
	\label{randomCond}
P\left(s_{i}=j,s_{i}\ne r_{i}\right) = P\left(r_{i}\in\left\{ 1,\dots,d_{1}-i_{\lam}\right\} \right)\cdot
	 P\left(s_{i}=j\left|\right.r_{i}\in\left\{ 1,\dots,d_{1}-i_{\lam}\right\} \right) 
	=\frac{d_{1}-i_{\lam}}{d_{1}}\cdot\frac{1}{i_{c}} .
\end{equation}

\subsection{Mirroring and Mirandom}
\label{sec:mirandom}

Using the above randomized uniform map $\vrp_u$ we have a way to define an FDR-controlling variant of our meta-procedure for
any pre-determined $c,\lam$. However, we can design more powerful procedures using alternative definitions of $\vrp$
(for the same values of $c,\lam$).

For example, with $c=\lam=1/2$ and an even $d_1$ we can consider, in addition to $\vrp_u$,
the mirror map: $\vrp_m(j)\coloneqq d_1-j+1$.
It is easy to see that under the conditions of Theorem \ref{NLCPthm},
$P\left(s_{i}=j,r_{i}\le d_{1}-i_{\lam}\right)=1/d_1$ hence \eqref{eq:decoy_win_pos_c_lam1} holds and the resulting method,
which we refer to as the ``mirror method''' (because when $L_i=-1$, $s_i$ is the rank symmetrically across the median to $r_i$),
controls the FDR. Similarly, we can choose to use a shift map $\vrp_s$: $\vrp_s(j)=j+d_1/2$, which will
result in a third FDR-controlling variant of our meta-procedure for $c=\lam=1/2$.

Comparing the shift and the mirror maps we note that when $L_i=-1$,
$\vrp_s$ replaces middling target scores with high decoy scores, whereas $\vrp_m$
replaces low target scores with high decoy scores.
Of course, the high decoy scores are the ones more likely to appear
in the numerator of \eqref{eq:reject_criterion-general-c-lam}, and generally
we expect the density of the target scores to monotonically decrease
with the quality of the score. Taken together, it follows that the
estimated FDR will generally be higher when using $\vrp_s$ than when using $\vrp_m$,
and hence the variant that uses $\vrp_s$ will be weaker than
the mirror. By extension the randomized $\vrp_u$ will fall somewhere between the other two maps,
as can be partly verified by the comparison of the power using $\vrp_m$ and $\vrp_u$ in panel A of \supfig~\ref{fig:power_initial}.

We can readily extend the mirroring principle to other values of $c$ and $\lam$ where $i_c$ divides $d_1-i_\lam$,
however when $i_c \nmid d_1-i_\lam$ we need to introduce some randomization into the map. Basically, we accomplish this
by respecting the mirror principle as much as we can while using the randomization to ensure that \eqref{eq:decoy_win_pos_c_lam1} holds ---
hence the name \emph{mirandom} for this map/procedure. It is best described by an example.

Suppose $d=7$. Then for  $i_{c}=3$ ($c=3/8$) and $i_{\lam}=4$ ($\lam=1/2$) the mirandom map, $\vrp_{md}$, is defined as
\[
\vrp_{md}(j)=\begin{cases}
8 & j=1\\
8 \text{ (with probability $1/3$), or $7$ (with probability $2/3$)} & j=2\\
7 \text{ (with probability $2/3$), or $6$ (with probability $1/3$)} & j=3\\
6  & j=4
\end{cases}
\]
Note the uniform coverage ($4/3$) of each value in the range, implying that if $j$ is randomly and uniformly chosen in the domain
then $\vrp_{md}(j)$ is uniformly distributed over $\{6,7,8\}$.

More generally the mirandom map $\vrp_{md}$ for a given $c\le\lam$ is defined in two steps.
In the first step it defines a sequence of $d_{1}-i_{\lam}$ distributions
$F_{1},\dots,F_{d_{1}-i_{\lam}}$ on the range $\left\{ d_{1}-i_{c}+1,\dots,d_{1}\right\} $
so that
\begin{itemize}
\item each $F_{l}$ is defined on a contiguous sequence of natural numbers, and
\item if $j<l$ then $F_{j}$ stochastically dominates $F_{l}$ and $\min\supp\left\{ F_{j}\right\} \ge\max\supp\left\{ F_{l}\right\} $.
\end{itemize}
In practice, it is straightforward to construct this sequence of distributions
and to see that, when combined, they necessarily satisfy the following equal coverage property:
for each $j\in\left\{ d_{1}-i_{c}+1,\dots,d_{1}\right\}$, 
$\sum_{l=1}^{d_{1}-i_{\lam}}F_{l}\left(j\right)=\frac{d_{1}-i_{\lam}}{i_{c}}$.
In the second step, mirandom defines $s_{i}$ for any $i$ with $r_{i}\in\left\{ 1,\dots,d_{1}-i_{c}\right\}$
by randomly drawing a number from $F_{r_{i}}$ (independently of everything
else).

It follows from the equal coverage property that for any $i\in N$ and $j\in\left\{ d_{1}-i_{c}+1,\dots,d_{1}\right\}$
\eqref{eq:decoy_win_pos_c_lam1} holds for $\vrp_{md}$ for essentially the same reason it held for $\vrp_u$ in \eqref{randomCond}.
Hence, the mirandom map allows us to controls the FDR for any pre-determined values of $c,\lam$.

\subsection{Data-driven setting of the tuning parameters $c,\lam$: finite-decoy Storey (FDS)\label{subsec:FDS}}

All the procedures we consider henceforth are based on the mirandom map. Where they differ is in how they set $c,\lam$.

For example, choosing $c=\lam=1/2$ gives us the mirror (assuming $d_1$ is even),
$c=\lam=1/d_1$ yields the max, while choosing $\lam=1/2$ and $c=\alp\le1/2$ coincides
with Lei and Fithian's recommendation in the related context of sequential hypothesis testing (technically we set $c=\lfloor\alp\cdot d_1\rfloor/d_1$ and refer to this method as ``LF'').
All of these seem plausible; however, our extensive simulations (\supsec~\ref{sec:Simulation-setup}) show that none dominates the others with substantial power
to be gained/lost for any particular problem (\supfig~\ref{fig:power_initial}, panels B-D).

As the optimal values of $c,\lam$ seem to vary in a non-trivial way with the nature of the data, as well as with
$d$ and $\alp$, we turned to data-driven approaches to setting $c,\lam$.

Lei and Fithian pointed out the connection between the $\left(c,\lam\right)$
parameters of their AS procedure (they refer to $c$ as $s$) and the corresponding parameters
in Storey's procedure. Specifically, AS's $\lam$ is analogous to
the parameter $\lam$ of \cite*{storey:strong} that determines
the interval $\left(\lam,1\right]$ from which $\pi_{0}$, the fraction
of true null hypotheses is estimated. Specifically, in the finite
sample case, $\pi_{0}$ is estimated as:
\begin{equation}
\hat{\pi}_{0}^{*}(\lam)=\frac{m-R(\lam)+1}{(1-\lam)m},\label{eq:pi0-est}
\end{equation}
 where $m$ is again the number of hypotheses, and $R(\lam)$ is the
number of discoveries at threshold $\lam$ (number of hypotheses whose
p-value is $\le\lam$). The $c$ parameter is analogous to the threshold
\begin{equation}
t_{\alpha}\left(\widehat{\text{FDR}}_{\lam}^{*}\right)=\sup\left\{ t\in\left[0,1\right]\,:\,\widehat{\text{FDR}}_{\lam}^{*}\le\alpha\right\} \label{eq:STS-c}
\end{equation}
 of \cite{storey:strong}, where 
\begin{equation}
\widehat{\text{FDR}}_{\lam}^{*}=\begin{cases}
\frac{m\cdot\hat{\pi}_{0}^{*}(\lam)\cdot t}{R(t)\vee1} & t\le\lam\\
1 & t>\lam
\end{cases}.\label{eq:STS-est-FDR}
\end{equation}

We take this analogy one step further and essentially use the 
procedure of \cite{storey:strong} to determine $c$ by applying
it to the empirical p-values, $\til p_{i}\coloneqq (d_1-r_{i}+1)/d_{1}$.
However, to do that, we first need to determine $\lam$.

We could have determined $\lam$ by applying the bootstrap approach of \cite{storey:strong} to $\til p_{i}$.
However, in practice we found that using the bootstrap option of the qvalue package \citep{storey:qvalueR} in our
setup can significantly compromise our FDR control.
Therefore, instead we devised an alternative approach inspired by the spline-based method of \cite{storey:statistical}
for estimating $\pi_{0}$, where we look for the flattening of the tail of the p-value
histogram as we approach 1. Because our p-values, $\til p_{i}$, lie on the lattice $i/d_{1}$ for $i=1,\dots,d_{1}$,
instead of threading a spline as in \cite{storey:statistical}, we repeatedly
test whether the number of p-values in the first half of the considered
tail interval $\left(\lam,1\right]$ is significantly larger than their number
in the second half of this interval (\supsec~\ref{subsec:Determining-lambda}).

Our finite-decoy Storey (FDS) procedure starts with determining $\lam$ as above
then essentially applies the methodology of \cite{storey:strong} to $\til p_{i}$ to set $c = t_{\alp}$
before applying mirandom with the chosen $c,\lam$.
Specifically, given the FDR threshold $\alp\in\left(0,1\right)$ and after determining $\lam$
FDS proceeds along (\ref{eq:pi0-est})-(\ref{eq:STS-est-FDR}) using
$R(\lam)=\left|\left\{ \til p_{i}\,:\,\til p_{i}\le\lam\right\} \right|$,
to determine
\begin{equation}
t_{\alpha}\left(\widehat{\text{FDR}}_{\lam}^{*}\right)=\max\left\{ i\in\left\{ 0,1,\dots,d_{1}\cdot\lam\right\} \,:\,\frac{m\cdot\hat{\pi}_{0}^{*}(\lam)\cdot i/d_{1}}{R\left(i/d_{1}\right)\vee1}\le\alp\right\} .\label{eq:t-FDS}
\end{equation}
This in principle is our threshold $c$ except that, especially when
$d$ is small, $t_{\alpha}\left(\widehat{\text{FDR}}_{\lam}^{*}\right)$
can often be zero which is not a valid value for $c$ in our setup.
Hence FDS defines 
\begin{equation}
c\coloneqq\max\left\{ 1/d_{1},t_{\alpha}\left(\widehat{\text{FDR}}_{\lam}^{*}\right)\right\} .\label{eq:c-FDS}
\end{equation}
With $\left(c,\lam\right)$ determined, FDS continues by applying
the mirandom procedure with the chosen parameter values.

FDS was defined as close as possible to \cite{storey:strong}'s
recommended procedure for guaranteed FDR control in the finite setting (albeit with a pre-determined $\lam$).
Indeed, as we argue in \supsec~\ref{suppsec:The-limiting-behavior}, FDS converges to a variant
of Storey's procedure once we let $d\lra\infty$ (the mirror and mirandom maps in general have an interesting limit in that setup).
However, we found that the following variant of FDS that we denote as FDS$_{1}$, often
yields better power in our setting, so we considered both variants.
FDS$_{1}$ differs from FDS as follows:
\begin{itemize}
	\item
	When computing $t_{\alpha}\left(\widehat{\text{FDR}}_{\lam}^{*}\right)$ \eqref{eq:t-FDS}
	we use Storey's large sample formulation which does not include the $+1$ in the estimator $\hat{\pi}_{0}^{*}(\lam)$ \eqref{eq:pi0-est}, and maximizes over $i\in\left\{ 0,1,\dots,d_{1}\right\}$.
\item Instead of defining $c$ as in (\ref{eq:c-FDS}), FDS$_{1}$ defines
$c\coloneqq\min\left\{ c_{\max},1/d_{1}+t_{\alpha}\left(\widehat{\text{FDR}}_{\lam}^{*}\right)\right\} $,
where $c_{\max}$ is some hard bound on $c$ (we used $c_{\max}=0.95$). 
\item With FDS$_1$'s tweaked definition of $c$ the case $c>\lam$ is now possible. However, mirandom
does not allow this so in that case FDS$_1$ sets $\lam\coloneqq c$ instead of the
value of $\lam$ defined above.
\end{itemize}

FDS and FDS$_1$ peek at the data to set $c,\lam$ hence they no longer fall
under mirandom's guaranteed FDR control and we resorted to empirically examining their
FDR control.
Specifically, we applied each method to 10K randomly drawn datasets for each of 1200 different
combinations of parameter values spanning a wide range of the parameter space (\supsec~\ref{sec:Simulation-setup}).
The 1200 combinations were evenly split between calibrated and non-calibrated scores
and the empirical FDR of a method at a selected FDR threshold $\alp$
is the 10K-sample average of the FDP of the method's discovery list at $\alp$.

As can be seen in panels A-C of \supfig~\ref{fig:FDR-control} the empirical violations of FDR control of FDS and FDS$_1$
are roughly in line with that of the max method.
Specifically, the overall maximal observed violation is 5.0\% for FDS, FDS$_1$ while it is 6.7\% for max.
Similarly, the number of curves (out of 1200) in which the maximal violation exceeds 2\% is 7 for FDS and FDS$_1$,
and 24 for the max. Given that the max provably controls the FDR these simulations suggest that
FDS and FDS$_1$ essentially control the FDR as well.

In terms of power FDS$_1$ seems to deliver overall more power than the mirror, max, LF, FDS and TDC, and often
substantially more (\supfig~\ref{fig:power_FDS1}, panels A-E).
We note, however, that at times FDS$_1$ has 10-20\% less power than the optimal method,
and we observed an even more extreme loss of power with the examples mentioned in \supsec~\ref{sec:Failures-of-established}
where BH has no power (Supplementary Section \ref{subsec:Failures-examples2}).
These issues motivate our next procedure.

\subsection{A bootstrap procedure for selecting an optimal method}

Our final, and ultimately our recommended multi-decoy procedure, uses a novel resampling approach
to choose the optimal procedure among several of the ones we introduced above.
Our optimization strategy is indirect: rather than using the resamples to choose the method
that maximizes the number of discoveries, we use the resamples to advise us whether or not such
a direct maximization approach is likely to control the FDR.

Clearly, a direct maximization would have been ideal had we been able to sample more instances
of the data. In reality, that is rarely possible all the more so with our underlying
assumption that the decoys are given and that it is forbiddingly expensive to generate additional ones.
Hence, when a hypothesis is resampled it comes with its original, as well as its decoy scores, thus
further limiting the variability of our resamples. In particular, direct maximization will occasionally
fail to control the FDR. Our Labeled Bootstrap monitored Maximization (LBM) procedure tries to identify those cases.

In order to gauge the rate of false discoveries we need labeled samples.
To this end, we propose a segmented resampling procedure that makes informed guesses (described below) about which
of the hypotheses are false nulls before resampling the indices.
The scores $\left\{ \tilde{Z}_{i}^{j}\right\} _{j=0}^{d}$ associated with each resampled conjectured
\emph{true null} index are then randomly permuted, which effectively
boils down to randomly sampling $j\in\left\{ 0,1,\dots,d\right\} $ and
swapping the corresponding original score $\tilde{Z}_{i}^{0}=Z_{i}$ with $\tilde{Z}_{i}^{j}$.

The effectiveness of our resampling scheme hinges on how informed are
our guesses of the false nulls.
To try and increase the overlap between our guesses and the true false nulls
we introduced two modifications to the naive approach of estimating the number
of false nulls in our sample and then uniformly drawing that many conjectured false nulls.
First, we consider increasing sets of hypotheses
$\text{\ensuremath{\mathcal{H}_{j}}}\subset\HC_{j+1}$ and verify
that the number of conjectured false nulls we draw from each $\HC_{j}$
agrees with our estimate of the number of false nulls in $\HC_{j}$.
Second, rather than being uniform,
our draws within each set $\HC_{j}$ are weighted according to the
empirical p-values so that hypotheses with more significant empirical
p-values are more likely to be drawn as conjectured false nulls.
Our segmented resampling procedure is described in detail in \supsec~\ref{subsec:resampling}.

In summary, LBM relies on the labeled resamples of our segmented resampling approach
to estimate whether we are likely to control the FDR when using direct maximization
(we chose FDS, mirror, and FDS$_1$ as the candidate methods). If so, then LBM uses the maximizing method;
otherwise, LBM chooses a pre-determined fall-back method (here we consistently use FDS$_{1}$).
Further details are provided in \supsec~\ref{supsec:LBM}.

Our simulations suggest that LBM's control of the FDR is on-par with that of the, provably FDR-controlling, max:
the overall maximal observed violation is 5.0\% for LBM while it is 6.7\% for max, and the number of curves (out of 1200)
in which the maximal violation exceeds 2\% is 21 for LBM, and 24 for the max (panels A and D, \supfig~\ref{fig:FDR-control}).
Power-wise LBM arguably offers the best balance among our proposed procedures, offering substantially
more power in many of the experiments while never giving up too much power when it is
not optimal (\supfig~\ref{fig:power_LBM}).
Finally, going back to the examples mentioned in \supsec~\ref{sec:Failures-of-established} we find that all our methods,
including LBM, essentially control the FDR where Storey's procedure substantially failed to do so,
and similarly that LBM delivers substantial power where BH had none (Supplementary Section \ref{subsec:Failures-examples2}).

\section{The peptide detection problem\label{sec:Peptide-ID-example}}

Our peptide detection procedure starts with a generalization of the WOTE procedure of \cite{granholm:determining}.
We use Tide~\citep{diament:faster} to find for each spectrum its best matching peptide
in the target database as well as in the $d$ decoy peptide databases.
We then assign to the $i$th target peptide the observed score, $Z_i$,
which is the maximum of all the PSM scores that were optimally matched to this peptide.
We similarly define the maximal scores of each of that peptide's $d$ randomly shuffled copies
as the corresponding decoy scores: $\til Z_{i}^{1},\dots,\til Z_{i}^{d}$.
If no spectrum was optimally matched to a peptide then that peptide's score is $-\infty$.

We then applied to the above scores TDC ($d=1$, with the $+1$ finite sample correction) --- representing a peptide-level
analogue of the picked target-decoy strategy of~\cite{savitski:scalable} --- as well as
the mirror, LBM and the averaging-based aTDC\footnote{We used the version named aTDC$_1^+$, which was empirically shown to control the FDR
even for small thresholds / datasets~\citep{keich:averaging}.} each using $d\in\left\{ 3,5,7,9\right\}$.

Clearly, the competition-based control of the FDR is subject to the
variability of the drawn competing scores. To ameliorate the effect of decoy-induced variability on our comparative analysis
we report the average of our analysis over 100 applications of each method 
using that many randomly drawn decoy sets (Supplementary Section \ref{subsec:real-data}).

We applied our analysis to three datasets: ``human'', ``yeast'' and ``ISB18''\footnote{
In the case of the ISB18 the data consists of 9 aliquots or replicates, so our analysis
was separately applied to each aliquot and then averaged over the aliquots.} (Supplementary Section \ref{subsec:real-data}),
and examined both the power of the considered methods (in all three datasets) and the FDR control (only in ISB18, as explained
in Supplementary Section \ref{subsec:real-data}).
Panel D of Figure \ref{fig:peptide-detect} suggests that when applied to the ISB18 dataset all our procedures seem to
control the FDR: the empirically estimated FDR is always below the selected threshold.
In terms of power, again we see that LBM is the overall winner: it typically delivers the largest
number of discoveries, and even in the couple of cases where it fails to do so  it is only marginally behind the top method (panels A--C).
In contrast, each of the other methods has some cases where it delivers noticeably fewer discoveries.
In practice, this means that scientists can extract more useful information (discoveries)
from the same wet lab experiment by leveraging more computational power coupled with LBM.

\begin{figure}
\centering %
\begin{tabular}{ll}
A: human & B: yeast\tabularnewline
\includegraphics[width=3in]{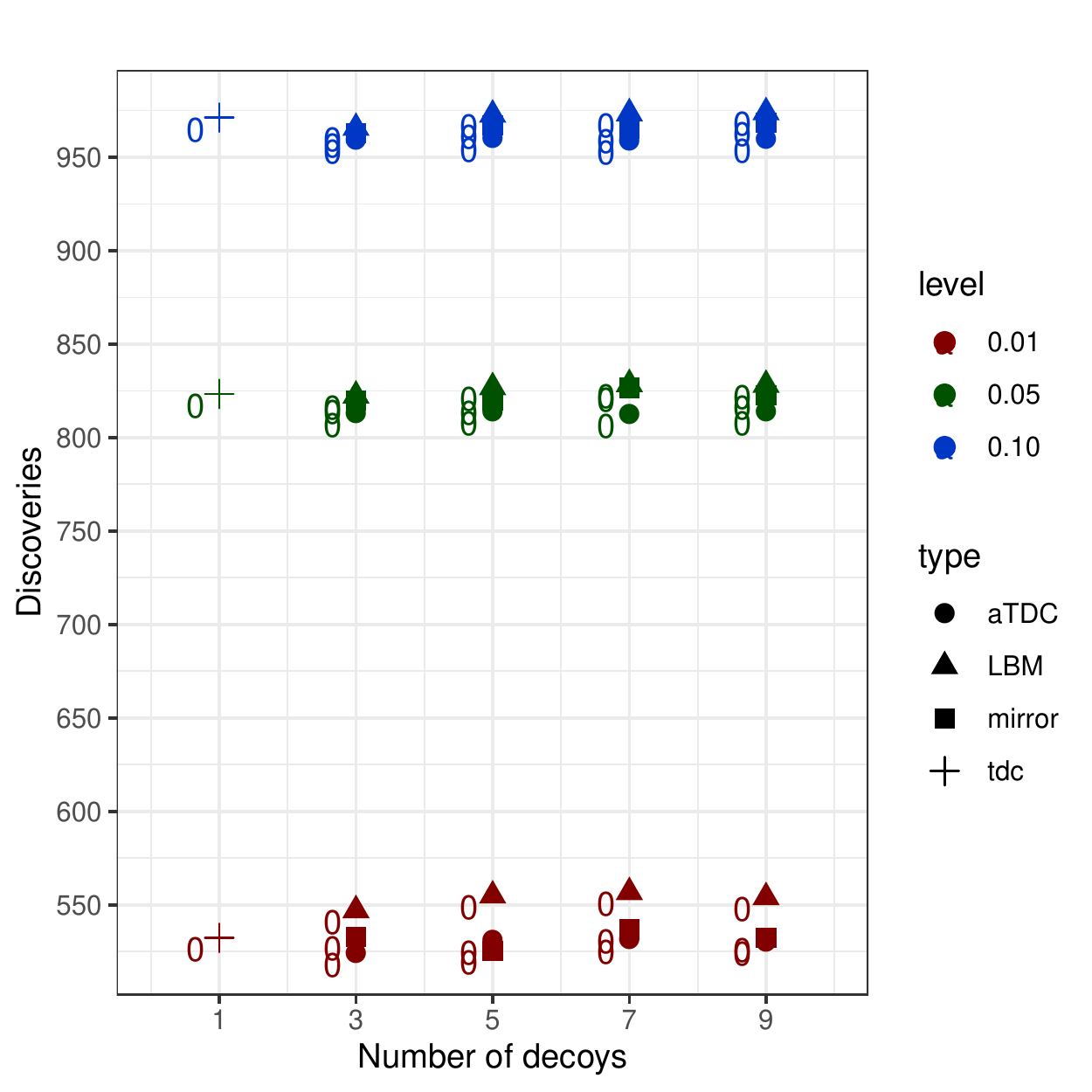} &
\includegraphics[width=3in]{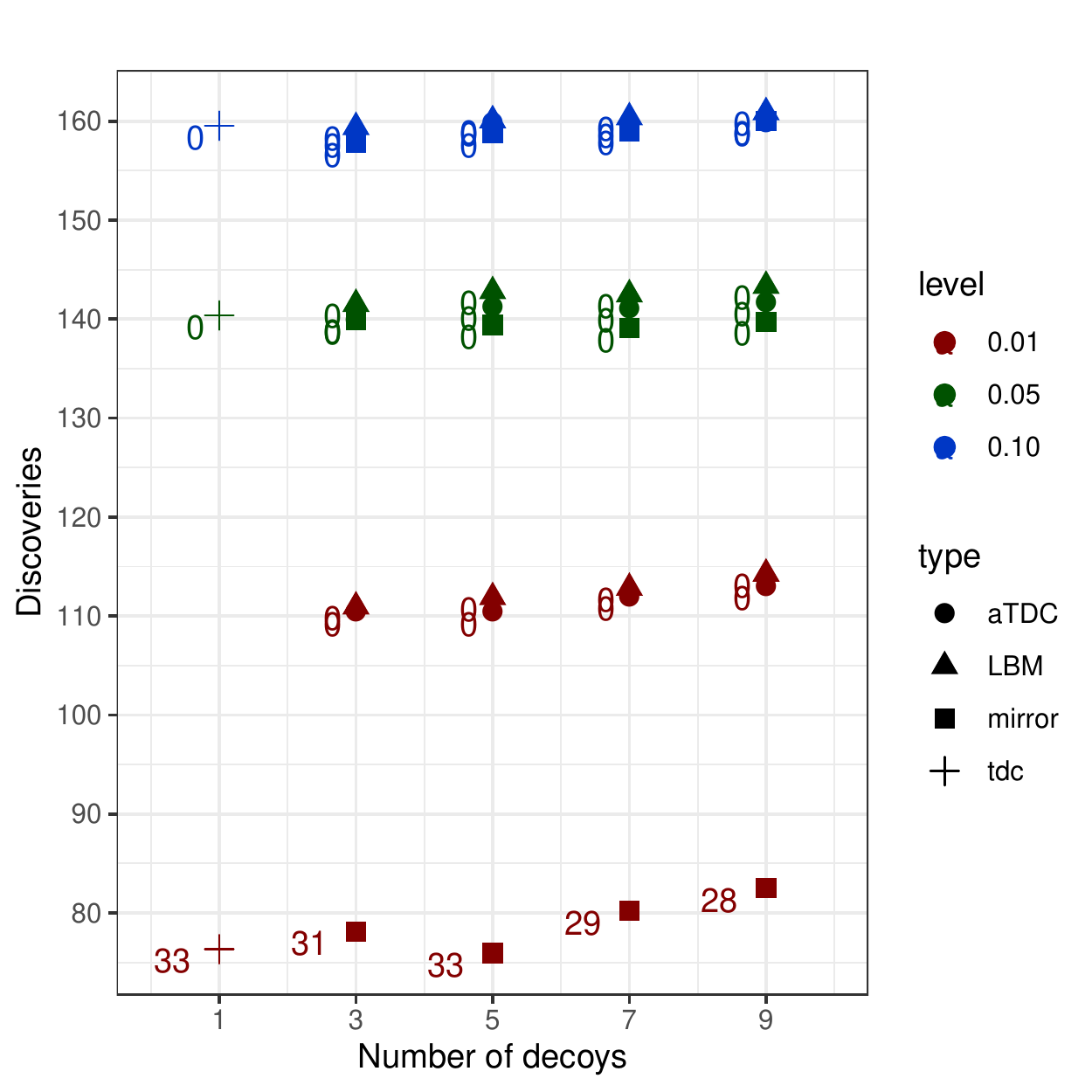}\tabularnewline
\tabularnewline
\tabularnewline
C: ISB18 (power) & D: ISB18 (FDR control)\tabularnewline
\includegraphics[width=3in]{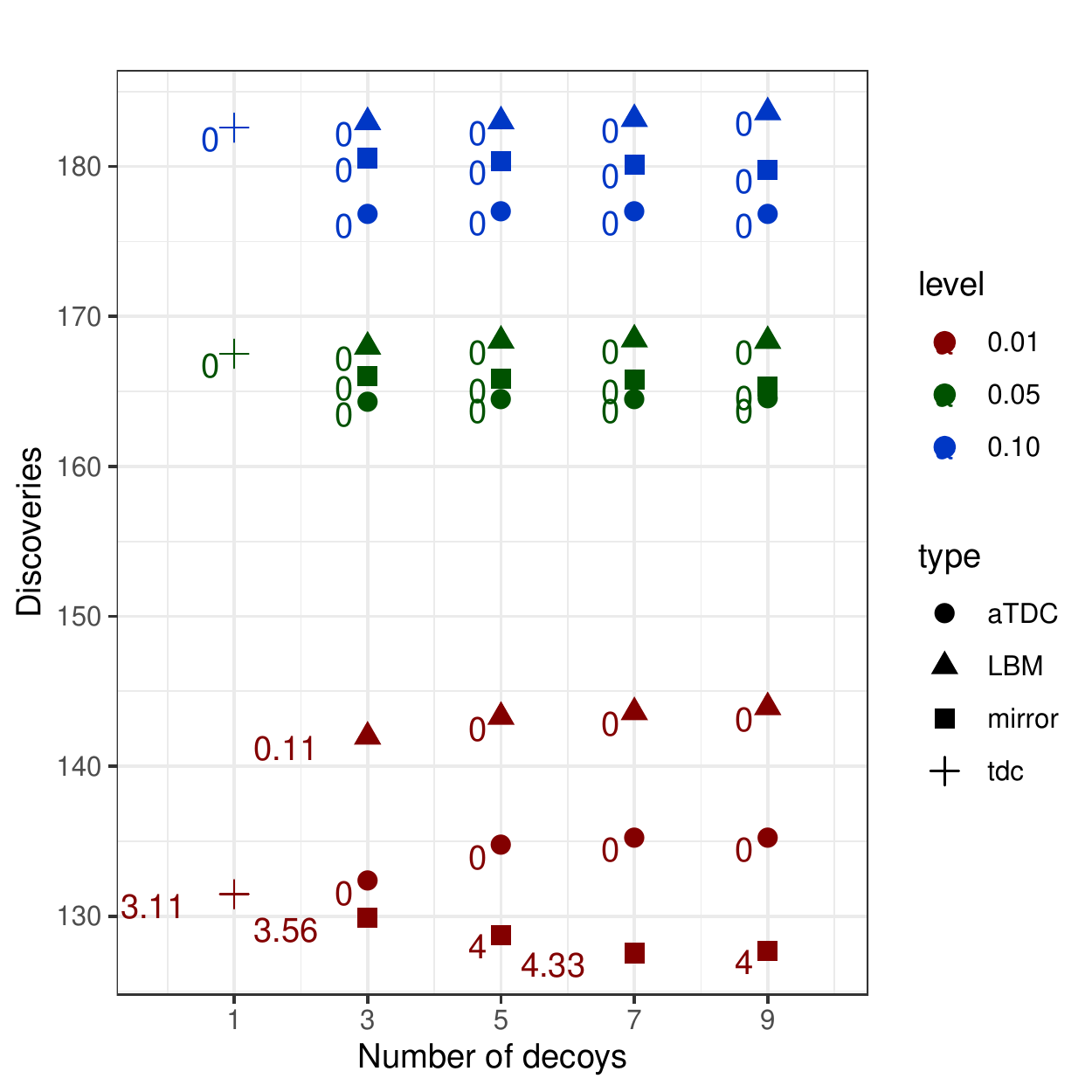} &
\includegraphics[width=3in]{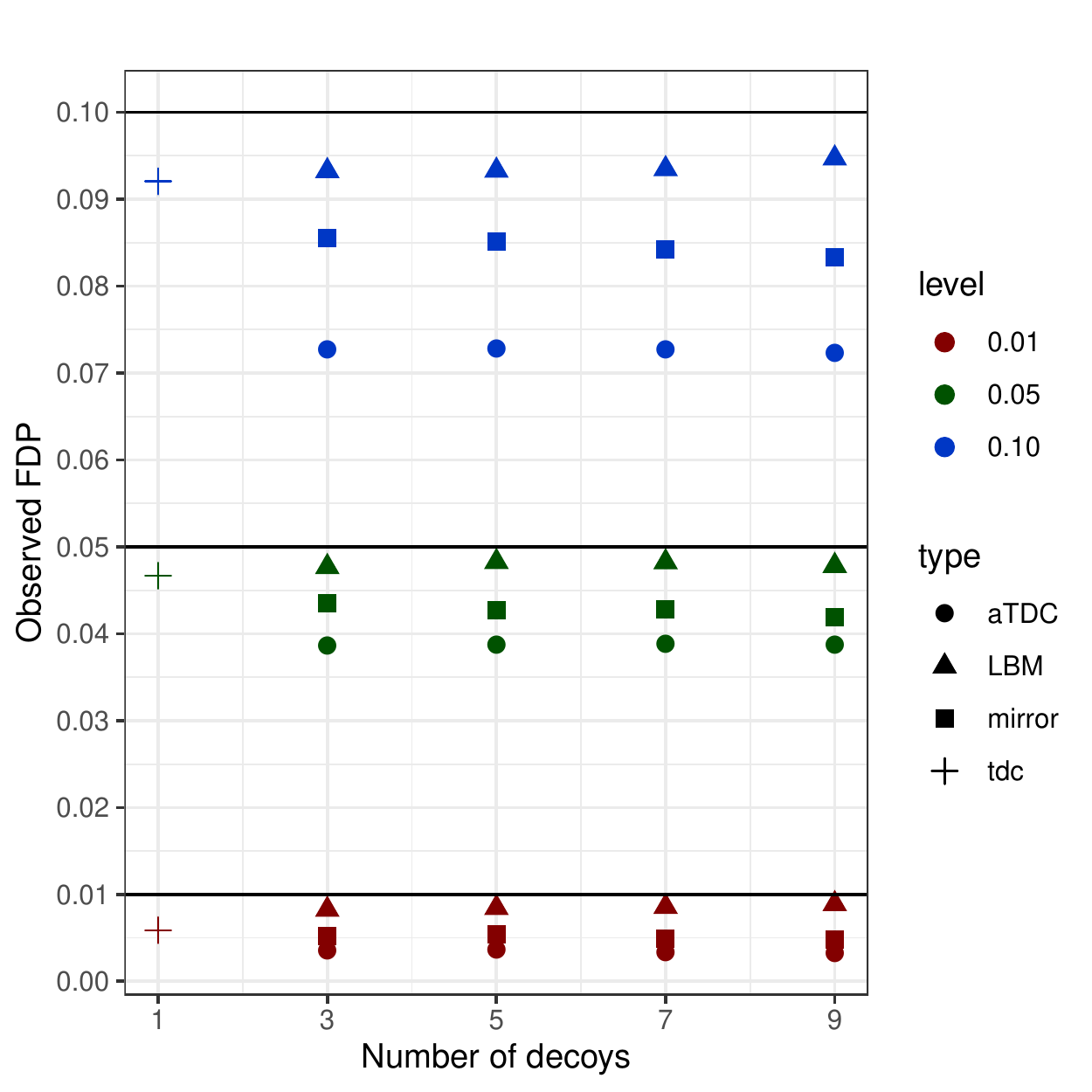}\tabularnewline
\end{tabular}\caption{\textbf{Peptide detection.}
\textbf{(A: human)} The number of discoveries at the given FDR threshold is the average over the 100 randomly drawn
decoys sets. Specifically, for each $d\in\left\{ 1,3,5,7,9\right\}$ we randomly drew 100 decoys
sets, each with exactly $d$ decoys, and applied
TDC ($d=1$), the mirror, LBM and aTDC ($d\in\left\{ 3,5,7,9\right\}$)
to the target and the drawn set of $d$ decoy scores. We then noted the number of target discoveries
at FDR thresholds of 1\%, 5\% and 10\%, and finally we averaged those numbers over the 100 drawn decoy sets.
The numbers to the left of the markers indicate the number of runs (out of 100) in which no discovery was reported.
\textbf{(B: yeast)} Same as (A) but for the yeast dataset.
\textbf{(C: ISB18, power)} Similar to (A) and (B) except for the ISB18 where the data consists of
9 aliquots or replicates. Therefore, the number of discoveries is averaged over the 9 aliquots,
where for each aliquot we averaged the number of discoveries over 100 randomly drawn sets of $d$ decoys as explained above.
The numbers to the left of the marker indicate the aliquots-average number of runs (out of 100) in which no discovery was reported.
\textbf{(D: ISB18, FDR)} Similar to (C) only here we noted the putative FDP of each run (as explained
in Supplementary Section \ref{subsec:real-data}),
then we averaged the FDP across the 100 runs to get an empirical FDR for each aliquot that we then averaged over the aliquots.
Notably, in all cases the empirical FDR was lower than the selected threshold.
\label{fig:peptide-detect}}
\end{figure}

More specifically, for $\alp=0.01$ LBM's average of 142.0 ISB18 discoveries ($d=3$) represents an 8.0\%
increase over TDC's average of 131.5 ISB18 discoveries, and we see a 9.4\% increase over TDC when using $d=5$ (143.3 discoveries).
In the human dataset and for the same $\alp=0.01$ we see a 2.8\% increase in power going from TDC to LBM with $d=3$ (532.4 vs.~547.1 discoveries),
and a 4.2\% increase when using LBM with $d=5$ (555.0 discoveries). LBM offers the biggest gains in the yeast dataset
where we see (again $\alp=0.01$) a 45.5\% increase in power going from TDC to LBM with $d=3$ (76.3 vs.~111.0 discoveries),
and a 46.7\% increase when using LBM with $d=5$ (111.9 discoveries).
Moreover, we note that for this $\alp=0.01$ TDC reported 0 yeast discoveries in 33 of the 100 runs (each using a different decoy
database), whereas LBM reported a positive number of discoveries in all 100 runs for each $d>1$ we considered.

At the higher FDR thresholds of $0.05$ and $0.1$ LBM offers a much smaller power advantage over TDC and is marginally behind
for $\alp=0.1$ and $d=3$ in the human and yeast datasets. Also, consistent with our simulations, we find that the mirror
lags behind LBM, and in fact in these real datasets it is roughly on par with TDC. 

To better understand the practical advantage offered by LBM, we
focused on a particular case: analysis of the yeast data set using an
FDR threshold of $\alpha = 0.01$.  The goal of many proteomics
experiments is to understand what biological pathways are active in a
given sample. We note that, as reported above, when using a single run of the yeast data
for 33 of the 100 decoy databases we drew, TDC found 0 discoveries at this 1\%
FDR threshold meaning that in 1/3 of similarly conducted experiments we would
not be able to draw any conclusion using this threshold. As noted above, LBM always reports some discoveries when
analyzing the same spectra set at 1\% FDR.

We next added two more spectra runs to the yeast dataset (Supplementary Section \ref{subsec:real-data}) representing
a higher budget experiment. In this case at 1\% FDR the average number of TDC discoveries was 275.9 and for LBM
using $d=5$ decoys it was 294. Accordingly, we focused on two sets of reported peptides, one of 294 peptides
detected by LBM (with $d=5$) and another of 276 peptides detected via TDC.
We eliminated from each group peptides that occur
in more than one protein, and then subjected the remaining 54 and 50
proteins, respectively, to analysis via the PANTHER Classification
System (\url{http://pantherdb.org}) \citep{mi:panther}.  Specifically,
we performed an overrepresentation test for Gene Ontology biological
process terms relative to the whole genome background.  We used the
``slim'' term set and Fisher's exact test, controlling FDR via BH at
5\%. This process yields 15 significantly overrepresented biological
process terms from the TDC list and 17 from the LBM list
(Table~\ref{table:panther}).  The two missing terms---''cellular
protein localization'' and ``cellular macromolecule
localization''---are closely related and imply that the sample under
investigation is enriched for proteins responsible in shuttling or
maintaining other proteins in their proper cellular compartments.
Critically, an analysis based solely on the traditional TDC approach would entirely miss this property of the sample being analyzed.

\begin{table}
  \begin{tabular}{lrr}
    \hline
    Gene Ontology term & TDC $q$-value & LBM $q$-value \\
    \hline
cellular response to unfolded protein (GO:0034620) & 9.35E-05 & 1.41E-04 \\
response to unfolded protein (GO:0006986) & 7.79E-05 & 1.18E-04 \\
chaperone-mediated protein folding (GO:0061077) & 1.37E-04 &2.07E-04 \\
cellular response to heat (GO:0034605) & 1.83E-04 & 2.77E-04 \\
response to heat (GO:0009408) & 1.99E-04 & 3.00E-04 \\
response to topologically incorrect protein (GO:0035966) & 3.11E-04 & 4.68E-04 \\
tRNA metabolic process (GO:0006399) & 2.42E-02 & 3.31E-02 \\
response to stress (GO:0006950) & 7.73E-03 & 1.15E-02 \\
cellular amino acid metabolic process (GO:0006520) & 1.14E-02 & 1.68E-02 \\
protein folding (GO:0006457) & 1.31E-02 & 1.93E-02 \\
translation (GO:0006412) & 9.74E-07 & 2.94E-06 \\
formation of translation initiation ternary complex (GO:0001677) & 1.37E-05 & 3.38E-05 \\
translational termination (GO:0006415) & 9.10E-06 & 2.26E-05 \\
translational elongation (GO:0006414) & 6.83E-06 & 1.69E-05 \\
cellular protein localization (GO:0034613) & --- & 2.93E-02 \\
cellular macromolecule localization (GO:0070727) & --- & 3.13E-02 \\
gene expression (GO:0010467) & 2.47E-02 & 4.77E-02 \\
    \hline
  \end{tabular}
  \caption{{\bf Statistical overrepresentation of Gene Ontology terms in the yeast data set.}  Each row is a Gene Ontology biological process term that is deemed significant at FDR $<0.05$.  Enrichments were tested twice, with respect to peptides identified using TDC and LBM.
    \label{table:panther}}
\end{table}

Finally, although aTDC was designed for the spectrum identification problem and in practice was never applied to the peptide detection
problem, it was instructive to add aTDC to this comparison.
LBM consistently delivered more detected peptides than aTDC did, although in some cases the
difference is marginal. Still, in the human dataset for $\alp=0.01$ with $d=3$ we see a 4.4\% increase in power going from
aTDC to LBM (524.2 vs.~547.1 discoveries), and with $d=5$ a 4.6\% increase when using LBM (530.8 vs.~555.0 discoveries).
Similarly, in the ISB18 dataset for $\alp=0.01$ with $d=3$ we see a 7.3\% increase in power going from
aTDC to LBM (132.3 vs.~142.0 discoveries), and with $d=5$ a 6.4\% increase when using LBM (134.7 vs.~143.3 discoveries).

\section{Discussion}

We consider a new perspective on the peptide detection problem, which can be framed
more broadly as multiple-competition based FDR control.
The problem we pose and the tools we offer can be viewed as bridging the gap between the canonical
FDR controlling procedures of BH and Storey and the single-decoy
approach of the popular TDC used in spectrum identification (ID).
Indeed, our proposed FDS converges to Storey's method as the number
of decoys $d\lra\infty$ (\supsec~\ref{suppsec:The-limiting-behavior}).

The methods we propose here rely on our novel mirandom procedure,
which guarantees FDR control in the finite sample case for any pre-determined values
of the tuning parameters $c,\lam$.
Our extensive simulations show that which of our methods delivers
the maximal power varies with the properties
of the experiment, as well as with the FDR threshold $\alp$.
This variation motivates our introduction of LBM. LBM relies on a novel labeled resampling technique,
which allows it to select its preferred method after testing whether a direct maximization
approach seems to control the FDR. Our simulations
suggest that LBM largely controls the FDR and seems to offer the best balance among our multi-decoy methods as well as a
significant power advantage over the single-decoy TDC.

Applying our methods to the peptide detection problem suggests
that, as with the simulated data, the FDR is controlled and LBM can
deliver significantly more power (up to almost 50\% more discoveries)
than the single decoy TDC.

We stress that the problem of peptide detection is important not only as a stepping stone toward the downstream goal of detecting proteins; in many proteomics studies, the peptides themselves are of primary interest.
For example, MS/MS is being increasingly applied to complex samples, ranging from the microbiome in the human gut \citep{lin:proteomics} to microbial communities in environmental samples such as soil or ocean water \citep{saito:progress}.
In these settings, the genome sequences of the species in the community are only partially characterized, so protein inference is problematic.
Nonetheless, observation of a particular peptide can often be used to infer the presence of a group of closely related species (a taxonomic clade) or closely related proteins (a homology group).
Peptide detection is also of primary interest in studies that aim to detect so-called ``proteoforms'' --- variants of the same protein that arise due to differential splicing of the mature RNA or due to post-translational modifications of the translated protein.
Identifying proteoforms can be critically important, for example, in the study of diseases like Alzheimer's or Parkinson's disease, in which the disease is hypothesized to arise in part due to the presence of deviant proteoforms \citep{morris:tau, ping:global, wildburger:diversity}.

Finally, as mentioned in the Introduction, our approach is applicable beyond peptide detection.
Moreover, while we stated our results assuming iid decoys, the results hold in a
more general setting of \emph{``conditional null exchangeability''} (\supsec~\ref{condExch}).
This exchangeability is particularly relevant for future work on generalizing the construction of \cite{barber:controlling}
to multiple knockoffs, where the iid decoys assumption is unlikely to hold.

\textbf{Related work.} We recently developed aTDC in the context of spectrum ID.
The goal of aTDC was to reduce the decoy-induced variability associated with TDC
by averaging a number of single-decoy competitions \citep{keich:progressive,keich:averaging}.
As such, aTDC fundamentally differs
from the methods of this paper which simultaneously use all the decoys in a single competition;
hence, the methods proposed here can deliver a significant power advantage over aTDC (panel F, \supfig~\ref{fig:power_LBM} and \supfig~\ref{fig:peptide-detect}).
Our new methods are designed for the iid (or exchangeable) decoys case, which is a reasonable assumption
for the peptide detection problem studied here but does not hold for the spectrum ID for which aTDC was devised.
Indeed, as pointed out in \cite{keich:controlling},
due to the different nature of native/foreign false discoveries, the
spectrum ID problem fundamentally differs from the setup of this paper
and even the above, weaker, null exchangeability property does not hold in this case.
Thus, LBM cannot replace aTDC entirely; indeed, LBM is too liberal in the context
of the spectrum ID problem. Note that in practice aTDC has not previously been applied to the peptide detection problem.

While working on this manuscript we became aware of a related Arxiv submission~\citep{he:direct}.
The initial version of that paper had just the mirror method, which as we show is quite limited
in power. A later version that essentially showed up simultaneously with the submission of our technical report \citep{emery:multiple}
extended their approach to a more general case; however,  the method still consists of a subset of our independently developed
research in that: (a) they do not consider
the $\lam$ tuning parameter, (b) they use the uniform random map $\vrp_u$ which, as we show, is inferior to mirandom, and (c) they
do not offer either a general deterministic (FDS) or bootstrap based (LBM) data-driven selection of the tuning parameter(s),
relying instead on a method that works only in the limited case-control scenario they consider.

{
\footnotesize  \bibliographystyle{plainnat} 
\bibliography{refs}
}

\clearpage

\section{Supplementary Material}

\subsection{The problem with pooling the decoys}\label{sec:Failures-of-established}

Two significant problems arise
when pooling the decoys to compute the p-values. First, these p-values
do not satisfy the assumption that the p-values
of the true null hypotheses are independent: because all p-values
are computed using the same batch of pooled decoy scores, it is clear
that they are dependent to some extent. While this dependency diminishes
as $m\ra\infty$, there is a second, more serious problem that in general
cannot be alleviated by considering a large enough $m$.

Specifically, in pooling the decoys we make the implicit assumption that the score
is calibrated, i.e., that all true null scores are generated according
to the same distribution. If this assumption is violated, as is typically
the case in the spectrum identification problem for one~\citep{keich:importance},
then the p-values of the true null hypotheses are not identically
distributed, and in particular they are also not (discrete) uniform
in general. This means that even the more conservative BH procedure
is no longer guaranteed to control the FDR, and the problem is much worse with Storey.
Indeed, Supplementary Section \ref{subsec:Failures-examples} below shows that there are
arbitrary large examples wherein Storey significantly fails to
control the FDR, and similar ones where BH is essentially powerless.
Those examples, demonstrate that,
in general, applying BH or Storey's procedure to p-values that are
estimated by pooling the competing null scores can be problematic
both in terms of power and control of the FDR.
Note that these issues have previously been discussed in the context of the spectrum identification problem, where the
effect of pooling on power and on FDR control were demonstrated using
simulated and real data~\citep*{keich:importance,keich:improved}.

\subsubsection{Examples of failings of the canonical procedures\label{subsec:Failures-examples}}

Consider BH applied to just $m=2$ hypotheses with $d=1$ decoy, and
suppose that $P\left(\til Z_{1}^{1}>\til Z_{2}^{1}\right)=1$ (i.e.,
the support of the null distribution corresponding to $H_{1}$ is
disjoint and to the right of the support of the null distribution
corresponding to $H_{2}$). Suppose further that both $H_{i}$ are
true nulls, so that the FDR coincides with the FWER (family-wise error
rate), which is the probability of at least one (false) discovery.
It is easy to see that in this case using the FDR threshold $\alp\coloneqq2/3$
the event $\left\{ Z_{1}>\til Z_{1}^{1},Z_{2}<\til Z_{2}^{1}\right\} $
will produce one discovery ($p_{1}\coloneqq\text{p-value}\left(Z_{1}\right)=1/3,p_{2}\coloneqq\text{p-value}\left(Z_{2}\right)=1$),
and the disjoint event $\left\{ Z_{2}>\til Z_{2}^{1}\right\} $ will
produce two discoveries ($p_{2}=2/3$). However, these events have
a total probability of $1/4+1/2=3/4$ so the FWER=FDR is $>\alp$
in this case.

This effect can be much more pronounced in the case of Storey's method.
Suppose that the null hypotheses split into two equal sized groups,
$A$ and $B$, where for every $i\in A$ and $j\in B$, $P\left(\til Z_{i}^{1}<\til Z_{j}^{1}\right)=1$.
Suppose further that all the hypotheses in $A$ are false nulls with
scores $Z_{i}$ satisfying $P\left(Z_{i}>\til Z_{i}^{1}\right)=1$,
and that all hypotheses in $B$ are true nulls. The decoy-pooled p-values
will be essentially no greater than 1/2. Hence, Storey's estimate
of $\pi_{0}$, $\hat{\pi}_{0}(\lam)=\frac{m-R(\lam)}{(1-\lam)m}$,
where $m$ is the number of hypotheses, and $R(\lam)$ is the number
of hypotheses whose p-value is $\le\lam$, will significantly underestimate
$\pi_{0}$. For example, if $\lambda\geq0.5$ then $\hat{\pi}_{0}=0$,
which in turn implies that essentially all null hypotheses will be
rejected at any FDR level $\alp$ and particularly for $\alp<1/2$,
while the actual FDP will clearly be $1/2$. Even if $\lambda$ is
chosen to better fit these p-values, e.g., $\lambda=0.25$, or the
set-up is changed slightly to allow some group $A$ p-values to be
null so $\hat{\pi}_{0}\neq0$, the procedure will still significantly
underestimate $\pi_{0}$ and thus underestimate the actual FDR.
\begin{example}
\label{ex:Storey-fail}As a specific example in the above vein we
constructed an experiment with $m=300$ and $d=5$ decoys where group
$A$'s true null distribution is \textbf{$N(0,1)$,} and group $B$'s
true null distribution is \textbf{$N(50,1)$}. We set all 150 hypotheses
in group $B$ and 50 of the 150 hypotheses in group $A$ to be true
null, and we generated observed scores by sampling from the appropriate
null distribution above. We next generated the observed scores for
the 100 false null hypotheses in group A by sampling from the same,
significantly shifted, $N(50,1)$ distribution that we used to generate
all observed scores of group $B$. All competing null (decoy) scores
were generated using the group's null distribution. In this setup
we chose to leave a third of group $A$ as true nulls so that approximately
$50$ of the p-values will exceed $1/2$ ensuring that $\hat{\pi}_{0}>0$.

We then computed the pooled p-values and applied Storey's FDR controlling
procedure, as presented in the package \texttt{qvalue} \citep{storey:qvalueR}
(with $\lambda$ chosen using the \texttt{bootstrap} option). This
experiment was repeated using 1,000 randomly drawn sets, noting each
time the real FDP at FDR thresholds of $\alp=0.1$ and $\alp=0.2$.
As expected in this setting, Storey's procedure clearly failed to
control the FDR: at $\alpha=0.1$, the empirical FDR (the FDP averaged
over the 1K samples) was $0.24$, or over $200\%$ of what it should
be, and for $\alpha=0.2$ the empirical FDR was larger than $0.5$
again indicating a significant violation.
\end{example}
We could not find such examples, with an essentially arbitrary large
$m$ and a substantial liberal bias, when using the BH procedure.
However, we found a class of arbitrary large examples, similar to
the above class (on which Storey fails to control the FDR), where
due to pooling the conservative nature of BH was amplified to the
point where it was essentially powerless. Consider four groups $A$,
$B$, $C$ and $D$ and suppose that for every $i\in A$, $j\in B$,
$k\in C$ and $l\in D$, $P\left(\til Z_{i}^{1}<\til Z_{j}^{1}<\til Z_{k}^{1}<\til Z_{l}^{1}\right)=1$.
Suppose further that all the hypotheses in groups $A$ and $B$ are
false null with scores that fall in the range of values of the subsequent
group, so in particular $P\left(Z_{i}>\til Z_{k}^{1}\right)=0$ and
similarly $P\left(Z_{j}>\til Z_{l}^{1}\right)=0$. It is easy to see
that using pooling in this case the p-values for the (false) null
hypotheses in groups $A$ and $B$ will be $\ge1/2$ and $1/4$ respectively,
and it follows that no discoveries can be made by BH with $\alpha<1/4$,
regardless of how large $m$ and $d$ are.
\begin{example}
\label{ex:BH-fail}Again, we construct a specific example according
to the above general outline. We set $m=300$, so that each of the
four groups has 75 hypotheses, and we use $d=5$ decoys. The null
distribution of each group is set as $N(\mu,1)$, where $\mu$ increases
from $\mu_{A}=0$, by 50, to $\mu_{D}=150$. The observed scores corresponding
to the 150 false null hypotheses of groups $A$ and $B$ were drawn
from the null distributions of group $B$ and $C$ respectively, whereas
the 150 observed scores of groups $C$ and $D$ were drawn from their
respective null distributions. Using pooled p-values BH does not yield
any discovery for any $\alpha\le0.65$ amongst any of our 1000 samples,
and it was not until using $\alpha=0.7$ that we finally started seeing
some samples on which BH had non-zero power. Incidentally, even using
non-pooled p-values is slightly better here: the first samples with
non-zero BH power appear for $\alpha=0.3$.
\end{example}

\subsection{Simulation setup\label{sec:Simulation-setup}}

In order to analyze and compare the performance of the FDR-controlling procedures we 
simulated datasets with both calibrated (all true null scores are generated according to the same distribution) and non-calibrated scores ---
a comparison that also allowed us to select our overall recommended procedure.

In the non-calibrated case we allow the distribution of the null scores to vary with
the hypotheses so we sample from hypothesis-specific distributions.
Specifically, for simulating using a non-calibrated score we associate
with the null hypothesis $H_{i}$ a normal $N(\mu_{i},\sig_{i}^{2})$
distribution from which its competing null (decoy/knockoff) scores
are sampled. If $H_{i}$ is labeled a true null, this is also the
distribution from which the observed score is sampled. Otherwise,
$H_{i}$ is a false null, so the observed score is sampled from a
$\gam_{i}$-shifted normal $N(\mu_{i}+\gam_{i},\sig_{i}^{2})$ distribution,
where $\gam_{i}>0$. The parameters $\mu_{i}$, $\sig_{i}^{2}$, and
$\gam_{i}$ are themselves sampled with each newly sampled set of
scores: 
\begin{itemize}
\item $\mu_{i}$ is sampled from a normal $N(\mu,\sigma^{2})$ distribution
with the hyper-parameters $\mu=0$ and $\sigma^{2}=1$. 
\item $\sigma_{i}^{2}$ is sampled from $1+\exp(\omega)$, where $\exp(\omega)$
is the exponential distribution with rate $\omega=1$. 
\item $\gamma_{i}$ is sampled from $1+\exp(\nu)$, where $\nu$ is a hyper-parameter
that determines the separation between the false and true null scores
\end{itemize}
When simulating using a calibrated score the parameters $\mu_{i},\sig_{i}$
and $\gam_{i}$ are kept constant.


In our non-calibrated score simulations we drew 10K random sets of
observed and competing null scores (each with its own randomly drawn
values of $\mu_{i},\sig_{i},\gam_{i}$) for each of the following
600 combinations of parameter (or hyper-parameter) values:
\begin{itemize}
\item The number of false null hypotheses, $k$, was set to each value in
$\left\{ 1,10,10^{2},10^{3},10^{4}\right\} $.
\item For each value of $k$, the total number of hypotheses, $m$, was
set to $\min\{c\cdot k, 2\cdot10^{4}\}$ where $c$ was set to each of the following factors
$\left\{ 1.25,2,4,10,20,100,1000\right\}$ subject to the constraint
that $m\ge100$.
\item For each values of $k$ and $m$ above, the hyper-parameter $\nu$
that determines the separation between the false and true null scores
was set to each of the values in \\$\{0.01,0.05,0.1,0.25,0.5,1.0\}$.
\item For each values of $k$, $m$, and $\nu$ above, the number of decoys
$d$ was set to each of the values in $\left\{ 3,5,9,19,39\right\} $.
\end{itemize}
We then used the 10K sampled sets from each of the 600 experiments
to find the empirical FDR as well as the power of each method for
each selected FDR threshold $\alp\in\Phi$. For a given threshold
$\alp$, the power of a method is the average percentage of false
nulls that are reported by the method at level $\alp$, and the empirical
FDR is the average of the FDP in the discovery list (both averages
are taken over the experiment's 10K runs). We used a fairly dense
set of FDR thresholds $\Phi$: from 0.001 to 0.009 by jumps of 0.001,
from 0.01 to 0.29 by jumps of 0.01, and from 0.3 to 0.95 by jumps
of 0.05.

Our calibrated score simulation also consisted of 600 experiments,
or combinations of parameter values. Specifically, we used the same
values of $k$, $m$, and $d$ as in the above non-calibrated simulations
and we let $\gam$ vary over the values in $\left\{ 0.8,1,1.4,2,3,4\right\} $.
In each experiment we again draw 10K random sets of observed and competing
scores using $\mu_{i}\equiv0,\sig_{i}\equiv1,\gam_{i}\equiv\gam$.

In both setups we examined the FDR control by looking at the ratio between the empirical FDR
(the observed FDP averaged over 10K runs) and the selected
threshold as well as at the power which is the average (over 10K runs) percentage of false nulls we discover.

\subsection{Determining $\protect\lam$ from the empirical p-values\label{subsec:Determining-lambda}}

Given an upper bound $\Lambda$ on $\lam$ (we used 0.95), and a binomial test significance cutoff $\beta$ (we used 0.1)
\begin{enumerate}
	\item Initialize: $i \coloneqq 1$.
	\item If $i \geq \Lambda \cdot d_1$ or $i = d$ then 
	\begin{itemize}
		\item set $i_\lambda \coloneqq i$ and stop.
	\end{itemize}
	\item If $i + d_1$ is even then
	\begin{itemize}
		\item set $i_s \coloneqq (i + d_1)/2$,
		\item $n_p^+ \coloneqq \#\left\{ \tilde{p}_i \in [(i_s + 1)/d_1, 1] \right\}$,
		\item $n_p^- \coloneqq \#\left\{\tilde{p}_i \in [ (i + 1)/d_1, i_s/d_1 ]\right\}$.
	\end{itemize}
	\item Otherwise,
	\begin{itemize}
		\item set $i_s \coloneqq (i + d_1 + 1)/2$,
		\item $n_p^+ \coloneqq \#\left\{\tilde{p}_i \in [(i_s + 1)/d_1, 1]\right\}$,
		\item $n_p^- \coloneqq \#\left\{\tilde{p}_i \in [ (i+1)/d_1, (i_s - 1)/d_1]\right\}$.
	\end{itemize}
	\item Calculate $p_b = P(B \ge n_p^-)$ where $B \sim Binomial(n_p^+ + n_p^-, 0.5).$
	\item If $p_b > \beta$ then (the remaining tail of the p-value histogram ``seems to have flattened'')
	\begin{itemize}
		\item set $i_\lambda \coloneqq i$ and stop.
	\end{itemize}
	\item Otherwise (we are yet to see the flattening of the tail of the p-value histogram),
	\begin{itemize}
		\item set $i \coloneqq i + 1$,
		\item return to step $2$.
	\end{itemize}
\end{enumerate}
Note that the interval $\left(\lam,1\right]$ from which $\pi_{0}$
is estimated in (\ref{eq:pi0-est}) coincides with $\left[(i_\lam+1)/d_{1},1\right]$.

\subsection{The limiting behavior of our FDR controlling methods\label{suppsec:The-limiting-behavior}}

In its selection of the parameter $c$, FDS essentially applies Storey's
procedure to the empirical p-values; however, there is a more intimate
connection between FDS, and more generally between some of the methods
described above and Storey's procedure that becomes clearer once we
let $d\ra\infty$. To elucidate that connection we further need to
assume here that the score is calibrated, that is, that the distribution
of the decoy scores is the same for all null hypotheses. In this case,
we might as well assume our observed scores are already expressed
as p-values: $Z_{i}=p_{i}$ (keeping in mind that this implies that
small scores are better, not worse as they are elsewhere in this paper).

It is not difficult to see that for a given $\left(c,\lam\right)$,
mirandom assigns, in the limit as $d\ra\infty$, $W_{i}\coloneqq Z_{i}=p_{i}$
if $p_{i}\le c$ ($L_{i}=1$, or original win), and $W_{i}\coloneqq\left(1-p_{i}\right)\cdot c/\left(1-\lam\right)\in\left[0,c\right)$
if $p_{i}>\lam$ ($L_{i}=-1$, or decoy win). Sorting the scores $W_{i}$
in increasing order (smaller scores are better here) $W_{(1)}<W_{(2)}<\dots<W_{(m)}$,
we note that for $i$ with $W_{\left(i\right)}=p_{\left(i\right)}\le c$
the denominator term $\#\left\{ j\le i\,:\,L_{(j)}=1\right\} $ in
(\ref{eq:reject_criterion-general-c-lam}) is the number of original
scores, or p-values $p_{j}\le p_{(i)}$. At the same time, for the
same $i$ and $j\le i$, $L_{(j)}=-1$ if and only if $p_{\left(j\right)}>\lam$
and $W_{\left(j\right)}<W_{\left(i\right)}=p_{\left(i\right)}\le c$
so we have for the numerator term
\[
\#\left\{ j\le i\,:\,L_{(j)}=-1\right\} =\#\left\{ j\,:\,p_{j}>\lam,W_{j}<p_{(i)}\right\} =\#\left\{ j\,:\,p_{j}>1-\frac{1-\lam}{c}p_{(i)}\right\} .
\]
Considering that $i_{\alp c\lam}<m$ in (\ref{eq:reject_criterion-general-c-lam})
must be attained at an $i$ for which $W_{i}=p_{i}\le c$ (original
win), we can essentially rewrite (\ref{eq:reject_criterion-general-c-lam})
as
\begin{equation}
i_{\alp c\lam}=\max\left\{ i\,:\,\frac{1+\#\left\{ j:\,p_{j}>1-\frac{1-\lam}{c}p_{(i)}\right\} }{\#\left\{ j:\,p_{j}\le p_{(i)}\right\} \vee1}\cdot\frac{c}{1-\lam}\le\alp\right\} .\label{eq:limit1}
\end{equation}

Consider now Storey's selection of the threshold $t_{\alp}$, which
when using the more rigorous estimate (\ref{eq:pi0-est}) essentially
amounts to
\[
t_{\alp}=\max\left\{ t\in\left[0,\lam^{*}\right]\,:\,\frac{1+\#\left\{ j:\,p_{j}>\lam^{*}\right\} }{\#\left\{ j:\,p_{j}\le t\right\} \vee1}\cdot\frac{t}{1-\lam^{*}}\le\alp\right\} ,
\]
 where $\lam^{*}$ is a tuning parameter. Considering the cases where
$t_{\alp}\le c$ and setting $\lam^{*}(t)\coloneqq1-\left(1-\lam\right)t/c$
Storey's threshold $t_{\alp}$ becomes
\[
t_{\alp}=\max\left\{ t\in\left[0,c\right]\,:\,\frac{1+\#\left\{ j:\,p_{j}>1-\frac{1-\lam}{c}t\right\} }{\#\left\{ j:\,p_{j}\le t\right\} \vee1}\cdot\frac{c}{1-\lam}\le\alp\right\} .
\]
Since in practice $t_{\alp}$ can be taken as equal to one of the
$p_{(i)}$ the equivalence with (\ref{eq:limit1}) becomes obvious
by identifying $t$ above with $p_{(i)}$ in (\ref{eq:limit1}).

Thus, for example, as $d\ra\infty$ the mirror method ($\lam=c=1/2$)
converges, in the context of a calibrated score, to Storey's procedure
using $\lam^{*}(t)\coloneqq1-t$, which coincides with the ``mirroring
method'' of~\cite{xia:neuralfdr}. It is worth noting that the general
setting $\lam^{*}(t)\coloneqq1-\left(1-\lam\right)t/c$ is not obviously
supported by the finite sample theory of~\cite{storey:strong}; however,
it can be justified by noting the above equivalence and our results
here.

An even more direct connection with Storey's procedure is established
by letting $d\ra\infty$ in the context of FDS. Indeed, using the
same $\lam$ determined by the progressive interval splitting procedure
described in Section \ref{subsec:Determining-lambda}, Storey's finite-sample
procedure (\ref{eq:STS-c}) would amount to setting the rejection
threshold to
\[
t_{\alp}=\max\left\{ t\in\left[0,\lam\right]\,:\,\frac{m\cdot\hat{\pi}_{0}^{*}(\lam)\cdot t}{R(t)\vee1}\le\alp\right\} .
\]
Recalling that $d_{1}\ra\infty$ we note that the latter $t_{\alp}$
coincides with the value FDS assigns to $c$ via (\ref{eq:t-FDS})
and (\ref{eq:c-FDS}). Let $i_{c}$ be such that the above $t_{\alp}=c\in\left[p_{\left(i_{c}\right)},p_{\left(i_{c}+1\right)}\right)$
(recall we assume no ties here), then we can assume without loss of
generality that $t_{\alp}=c=p_{\left(i_{c}\right)}$ and hence (compare
with (\ref{eq:limit1})) the mirandom part of FDS will find
\[
i_{\alp c\lam}=\max\left\{ i\le i_{c}\,:\,\frac{1+\#\left\{ j:\,p_{j}>1-\frac{1-\lam}{c}p_{(i)}\right\} }{\#\left\{ j:\,p_{j}\le p_{(i)}\right\} \vee1}\cdot\frac{c}{1-\lam}\le\alp\right\} .
\]
But
\[
\frac{1+\#\left\{ j:\,p_{j}>1-\frac{1-\lam}{c}p_{(i_{c})}\right\} }{\#\left\{ j:\,p_{j}\le p_{(i_{c})}\right\} \vee1}\cdot\frac{c}{1-\lam}=\frac{1+\#\left\{ j:\,p_{j}>\lam\right\} }{\#\left\{ j:\,p_{j}\le c\right\} \vee1}\cdot\frac{c}{1-\lam}=\frac{m\cdot\hat{\pi}_{0}^{*}(\lam)\cdot t_{\alp}}{R(t_{\alp})\vee1}\le\alp.
\]
Hence $i_{c}$ satisfies the required inequality and $i_{\alp c\lam}=i_{c}$.
It follows that the rejection lists of FDS and the above variant of
Storey's procedure with the same $\lam$ coincide in the $d\ra\infty$
limit.

In Supplementary Section 8.7.2 of our technical report \cite{emery:multiple} we compare
the limiting methods of our procedures as $d\ra\infty$
using the same kind of simulations we use in this paper for the finite decoys case.
These experiments showed trends similar to those found in the finite case, indicating that the relationships between the methods continue to the limit.

\subsection{Labeled resampling\label{subsec:resampling}}

For clarity
of the exposition we break the description of our labeled resampling procedure into two parts with the first describing
how we generate a sample of conjectured true/false null labels.
\begin{enumerate}
\item Determine $\lam$ as described in \supsec~\ref{subsec:Determining-lambda}
\item Using $c=\lam$ from step 1 above, apply steps 1-2 of mirandom (Section~\ref{sec:metaProc}
with $\vrp\equiv\vrp_{md}$ of Section~\ref{sec:mirandom}) to define the assigned scores $W_{i}$ and
labels $L_{i}$, and order the hypotheses in a decreasing order of $W_{i}$
\item Initialize by setting:
\begin{itemize}
\item $j\coloneqq1$ ($j$ is the index of the set of hypotheses we currently
consider)
\item $i_{1}\coloneqq1$, $i_{0}\coloneqq0$ ($i_{j}$ is the number of
hypotheses in $\HC_{j}$) 
\item $l\coloneqq0$ ($l$ denotes the index of last drawn conjectured false
null)
\item $\fb\coloneqq\left(0,0,\dots,0\right)$ ($\fb_{i}$ is the indicator
of whether or not we conjecture $H_{i}$ is a false null)
\end{itemize}
\item Estimate $a_{j}$, the number of false null hypotheses in $\HC_{j}=\left\{ H_{i}\,:\,i\le i_{j}\right\} $,
as
\[
a_{j}\coloneqq\left(\#\left\{ i\le i_{j}\,:\,L_{i}=1\right\} -\#\left\{ i\le i_{j}\,:\,L_{i}=-1\right\} \cdot\frac{\lam}{1-\lam}\right)\vee0.
\]
Note that the first term is the number of original wins among the
hypotheses in $\HC_{j}$ and the second is essentially the numerator
of (\ref{eq:reject_criterion-general-c-lam}) (with $c=\lam$), which uses the number of
decoy wins to estimate the number of false discoveries among those
original wins.
\item If $a_{j}>\left\Vert \fb\right\Vert _{1}$ (the number of conjectured
false nulls drawn so far) then draw $a_{j}-\left\Vert \fb\right\Vert _{1}$
additional conjectured false nulls as follows:
\begin{enumerate}
\item for each $i\in\left\{ l+1,l+2,\dots,i_{j}\right\} $ let $w_{i}\coloneqq1-\til p_{i}$,
where $\til p_{i}$ are the empirical p-values
\item while $a_{j}-\left\Vert \fb\right\Vert _{1}>0$:
\begin{enumerate}
\item draw an index $i\in\left\{ l+1,\dots,i_{j}\right\} $ according to
the categorical distribution with a probability mass function proportional
to $w_{i}$ 
\item set $\fb_{i}\coloneqq1$ and $w_{i}\coloneqq0$
\end{enumerate}
\end{enumerate}
\item If $i_{j}=m$ return the conjectured labels $\fb$, else continue
\item Set $\del_{j+1}\coloneqq i_{j}-i_{j-1}+1$ if no new conjectured false
null were drawn in step 5, otherwise set $\del_{j+1}\coloneqq i_{j}-i_{j-1}$
\item Set $i_{j+1}\coloneqq\left(i_{j}+\del_{j+1}\right)\wedge m$
\item Set $j\coloneqq j+1$ and go back to step 4
\end{enumerate}
Note that step 7 lets the data determine the number of hypotheses
in $\HC_{j+1}\setminus\HC_{j}$: this number grows if going from $\HC_{j-1}$
to $\HC_{j}$ we concluded we do not need to draw any additional conjectured
false nulls. This scheme is well adapted to handle a fairly common
scenario where most of the highest scoring hypotheses are false null,
making sure they will be labeled as such in our resamples.

The second phase of the algorithm simply resamples the indices in
the usual bootstrap manner and then randomly permutes the conjectured
true null scores:
\begin{enumerate}
\item independently sample $m$ indices $j_{1},\dots,j_{m}\in\left\{ 1,2,\dots,m\right\} $
\item for $i=1,\dots,m$: 
\begin{enumerate}
\item if $\fb_{j_{i}}=0$ draw a permutation $\pii_{i}\in\Pi_{d_{1}}$,
else, $\fb_{j_{i}}=1$ so define $\pi_{i}\coloneqq Id\in\Pi_{d_{1}}$
(the identity permutation)
\item apply the permutation $\pi_{i}$ to $\VV_{i}\coloneqq\left(\til Z_{j_{i}}^{0}\coloneqq Z_{j_{i}},\til Z_{j_{i}}^{1},\dots,\til Z_{j_{i}}^{d}\right)$:
$\VV_{i}\circ\pi_{i}\coloneqq\left(\til Z_{j_{i}}^{\pi_{i}(1)-1},\dots,\til Z_{j_{i}}^{\pi_{i}(d_{1})-1}\right)$
\end{enumerate}
\item return the resampled labeled data $\left\{ \left(\VV_{i}\circ\pi_{i},\fb_{j_{i}}\right)\,:\,i=1,\dots,m\right\} $
\end{enumerate}

\subsection{Labeled Bootstrap monitored Maximization (LBM)}
\label{supsec:LBM}

Given the list of original and decoy scores, an ordered list of candidate methods $\M$, a fall-back method $M_f$,
a set of considered FDR thresholds $\Phi$, and the number of bootstrap samples $n_b$, LBM executes the following steps:
\begin{enumerate}
	\item For each bootstrap/resample run $i = 1, \dots n_b$:
	\begin{enumerate}
		\item Generate a labeled resample as describe in \supsec~\ref{subsec:resampling} above.
		\item Apply each method $M\in \M$ to the resample noting the number of discoveries $D_M^i(\alp)$ for each $\alp\in\Phi$, as well as
		the corresponding FDP, $F_M^i(\alp)$ (computed based on the conjectured labels of the resample).
		\item For each $\alp\in\Phi$ sort the methods according to $D_M^i(\alp)$ with ties broken according to the rank of the methods in the list $\M$, and
		\begin{enumerate}
			\item
			record the rank $r_M^i(\alp)$ of each method,
			\item
			record $F_*^i(\alp)\coloneqq F_M^i(\alp)$ where $M$ is the highest rank method (with the largest number of discoveries).
		\end{enumerate}
	\end{enumerate}
	\item For each $\alp\in\Phi$:
	\begin{enumerate}
		\item
		Estimate the FDR of the direct maximization approach as the simple average 
			$\what{FDR}_*(\alp)\coloneqq \frac{1}{n_b}\sum_{i=1}^{n_b}F_*^i(\alp)$.
		\item If $\what{FDR}_*(\alp) > \alpha$ (see the comment below) then
		\begin{itemize}
			\item (the FDR of direct maximization seems too high so) set the selected method for this $\alp$ to the fall-back method: $S(\alp)\coloneqq M_f$.
		\end{itemize}
		Otherwise,
		\begin{itemize}
			\item (direct maximization seems to work fine so) set the selected method to the one with the highest average rank: $S(\alp)\coloneqq \argmax_M \sum_{i=1}^{n_b}r_M^i(\alp)$ (ties are broken according to the rank of the methods in the list $\M$).
		\end{itemize}
	\end{enumerate}
\end{enumerate}

We added to LBM two options that in practice were used throughout our reported simulations.
The first is that we allowed some slack when comparing $\what{FDR}_*(\alp)$ with $\alpha$ to
check whether the FDR of direct maximization seems too high (step 2b above).
Specifically, particularly because the empirical mean is taken over
a relatively small number of resamples (we used $n_{b}=50$ in our
applications), we instead checked whether
\[
\what{FDR}_*(\alp) > \alp + 4\sig(\alp)\cdot\left(1-\hat{\pi}_{0}^{*}(\lam)\right) ,
\]
where $\hat{\pi}_{0}^{*}(\lam)$ is the $\pi_{0}$ estimate
used by FDS$_{1}$ described in \supsec s~\ref{subsec:Determining-lambda}
and \ref{subsec:FDS}, and $\sig(\alp)$ is the estimated standard error of $\what{FDR}_*(\alp)$.
In practice, this relaxation lead to some increase in power with no visible impact on the FDR control.

The second option is a post-processing step that aims to produce
a monotone list of discoveries as a function of the FDR threshold $\alp$.
Specifically, we check if the number of discoveries at $\alp_{j+1}$ is smaller than the
number we have when using $\alp_{j}<\alp_{j+1}$, and if that is the
case, then we override our resampling-based selection of the optimal
method for $\alp_{j+1}$ and instead we use the same method that was
previously selected for $\alp_{j}$, i.e., $S(\alp_{j+1}) \coloneqq S(\alp_j)$.


In terms of the list of candidate methods we consider, $\M$, we need to strike a
balance between considering more methods, equivalently more choices
of $\left(c,\lam\right)$, and the increasing likelihood that the
fall-back would be triggered. In practice, we found that considering
the methods of FDS$_{1}$, mirror, and FDS (and in that the tie-breaking order, so FDS has the highest
priority) works well so the reported version of LBM uses this particular list of methods.

\subsection{Revisiting the failings of the canonical procedures\label{subsec:Failures-examples2}}

Going back to the two examples of Supplementary Section \ref{subsec:Failures-examples}
we note that all our methods essentially control the FDR with the
empirical FDR (FDP averaged over the 1K sets samples sets) below the
selected FDR threshold for all $\alp\in\Phi$ with a single exception
in Example \ref{ex:Storey-fail} at $\alp=0.2$, where FDS, FDS$_{1}$,
and LBM have an empirical FDR of 0.208: 4\% over the threshold, compared
with the $>250\%$ violation of Storey with pooled p-values.

Interestingly, when comparing the power of our methods in Example
\ref{ex:BH-fail}, where BH applied to the pooled p-values made no
true discoveries even at $\alp=0.65$, we find that both the mirror
and FDS$_{1}$ are significantly weaker than FDS, LF and LBM, again
demonstrating the utility of LBM. Specifically, at $\alp=0.15$ both
FDS's and LBM's power stand at 78.5\% and LF's at 62.8\% compared
with 0\% power for both the mirror and FDS$_{1}$. At $\alp=0.2$
FDS, LBM, and LF boast 100\% power while the mirror power stands at
0.1\% and FDS$_{1}$'s power is 0.9\%.

\subsection{Analysis of real data\label{subsec:real-data}}

We applied our analysis to three datasets.

The human data set consists of a single control run ({\footnotesize \texttt{CTL\_R1\_1}} from the data set with MassIVE identifier
{\footnotesize \texttt{MSV000079437}}~\citep{zhong:quantitative}.
The data was generated on an LTQ-Orbitrap Velos Pro on proteins extracted from human SH-SY5Y cells treated with 200 $\mu$M H$_2$O$_2$.
The human reference proteome was downloaded from Uniprot on 28 Nov 2016.

The yeast data set was analyzed twice. The first, using 
a single run ({\footnotesize \texttt{Yeast\_In-gel\_digest\_2}}) selected at random
from the data set with PRIDE identifier {\footnotesize \texttt{PXD002726}}~\citep{schittmayer:cleaning}.
This is the yeast analysis that is reported in Figure \ref{fig:peptide-detect}.
Because TDC reported 0 discoveries for 33 of the 100 decoys we drew, we performed a second
analysis of the yeast data, this time using all three runs
from the same PRIDE dataset.
The data was generated on an LTQ Orbitrap Velos on proteins extracted
from an in-gel digest of \emph{S. cerevisiae} lysate. The yeast reference
proteome was downloaded from Uniprot on 28 Nov 2016.

The ISB18 data set is derived from a series of experiments using an
18-protein standard protein mixture (\href{https://regis-web.systemsbiology.net/PublicDatasets} {https://regis-web.systemsbiology.net/PublicDatasets},~\citep{klimek:standard}).
We use 10 runs carried out on an Orbitrap (\texttt{Mix\_7}). The database
consists of the 18 proteins from the standard mixture, augmented with
the full \emph{H. influenzae} proteome, as provided by Klimek et al.

Searches were carried out using the Tide search engine \citep{diament:faster}
as implemented in Crux~\citep{park:rapid}. The peptide database included
fully tryptic peptides, with a static modification for cysteine carbamidomethylation
(\texttt{C+57.0214}) and a variable modification allowing up to six
oxidized methionines (\texttt{6M+15.9949}). Precursor window size
was selected automatically with Param-Medic~\citep{may:param-medic}.
The XCorr score function was employed for uncalibrated searches, using
a fragment bin size selected by Param-Medic.

Clearly, the competition-based control of the FDR is subject to the
variability of the drawn competing scores. To ameliorate this variability
here, we initially searched the spectra against 100 randomly shuffled
decoy databases, and then for each $d\in\left\{ 1,3,5,7,9\right\} $
we repeated our analysis drawing 100 sets, each with $d$ of those decoy databases, while
making sure that the 100 drawn sets are distinct. We can then compare
the number of discoveries reported by each considered method at the
selected FDR threshold $\alp$ (here $\alp\in\left\{ 0.01,0.05,0.1\right\} $).
More precisely, for each number of decoys $d$ we average the number
of discoveries over the 100 randomly drawn sets of $d$ decoys.

The ISB18 is a fairly unusual dataset in that it was generated using
a controlled experiment, so the peptides that generated the spectra
could have essentially only come from the 18 purified proteins used
in the experiment. We used this dataset to get some feedback on how well our
methods control the FDR, as explained next.

The spectra set was scanned against a target database that included,
in addition to the 463 peptides of the 18 purified proteins, 29,379
peptides of 1,709 \emph{H.~influenzae} proteins (with ID's beginning
with \texttt{gi|}). The latter foreign peptides were added in order
to help us identify false positives: any foreign peptide reported
is clearly a false discovery. Moreover, because the foreign peptides
represent the overwhelming majority of the peptides in the target
database (a ratio of 63.5 : 1), a native ISB18 peptide reported is
most likely a true discovery (a randomly discovered peptide is much
more likely to belong to the foreign majority). Taken together, this setup
allows us to gauge the actual FDP for each set of $d$ drawn decoys,
FDR threshold $\alp$, and the FDR controlling procedure that generated
the discovery list. More precisely, we average the FDP over
the 100 drawn sets of $d$ decoys.

The 87,549 spectra of the ISB18 dataset were assembled from 10 different
aliquots, so in practice we essentially have 10 independent
replicates of the experiment. However, the last aliquot had only 325
spectra that registered any match against the combined target database,
compared with an average of over 3,800 spectra for the other 9 aliquots,
so we left it out when we independently applied our analysis to each
of the replicates. By averaging the above decoy-drawn averaged FDP
over the 9 aliquots we obtain an estimate of the FDR that we can
compare to the selected FDR threshold.

Similarly, when gauging the power of a method on the ISB18 dataset
our analysis was separately applied to each aliquot and then averaged over the aliquots.

\subsection{Conditional null exchangeability}
\label{condExch}

The following condition which is a generalization of the conditional exchangeability property of \cite{barber:controlling}
is weaker than the iid decoys condition. Nevertheless, it can be shown that it is sufficient for our statements to hold.
\begin{defn}
\label{def:cond_exch} Let $\VV_{i}\coloneqq\left(\til Z_{i}^{0},\til Z_{i}^{1},\dots,\til Z_{i}^{d}\right)$,
where $\til Z_{i}^{0}\coloneqq Z_{i}$ is the $i$th original score and $\til Z_i^1,\dots,\til Z_i^d$ are the corresponding $d$ decoy scores,
and let $\Pi_{d_{1}}$ denote the set of all permutations on $\left\{ 1,\dots,d,d+1\eqqcolon d_{1}\right\} $.
With $\pi\in\Pi_{d_{1}}$ we define $\VV_{i}\circ\pi\coloneqq\left(\til Z_{i}^{\pi(1)-1},\dots,\til Z_{i}^{\pi(d_{1})-1}\right)$,
i.e., the permutation $\pi$ is applied to the indices of the vector
$\VV_{i}$ rearranging the order of its entries. Let $N\subset\left\{ 1,2,\dots,m\right\} $
be the indices of the true null hypotheses and call a sequence of
permutations $\pi_{1},\dots,\pi_{m}$ with $\pi_{i}\in\Pi_{d_{1}}$
a \emph{null-only sequence} if $\pi_{i}=Id$ (the identity permutation)
for all $i\notin N$. We say the data satisfies the \emph{conditional
null exchangeability property} if for any null-only sequence of permutations
$\pi_{1},\dots,\pi_{m}$, the joint distribution of $\VV_{1}\circ\pi_{1},\dots,\VV_{m}\circ\pi_{m}$
is invariant of $\pi_{1},\dots,\pi_{m}$. 
\end{defn}


\subsection{Figures}

\begin{figure}
\centering %
\begin{tabular}{ll}
A: randomized ($\vrp_u$) vs.~mirror ($\vrp_m$); $\lam=c=1/2$  & B: max vs.~mirror \tabularnewline
\includegraphics[width=3in]{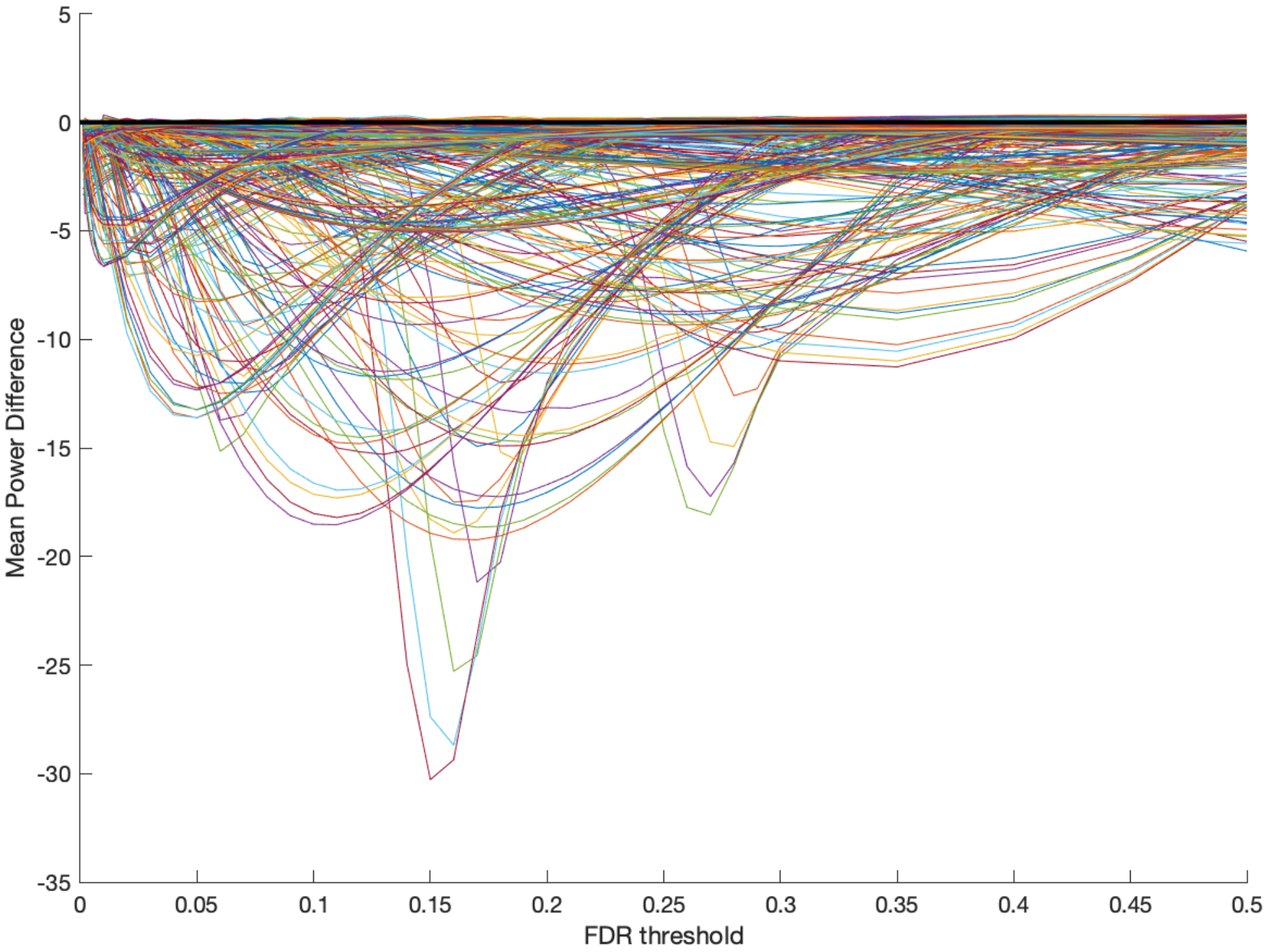}  & \includegraphics[width=3in]{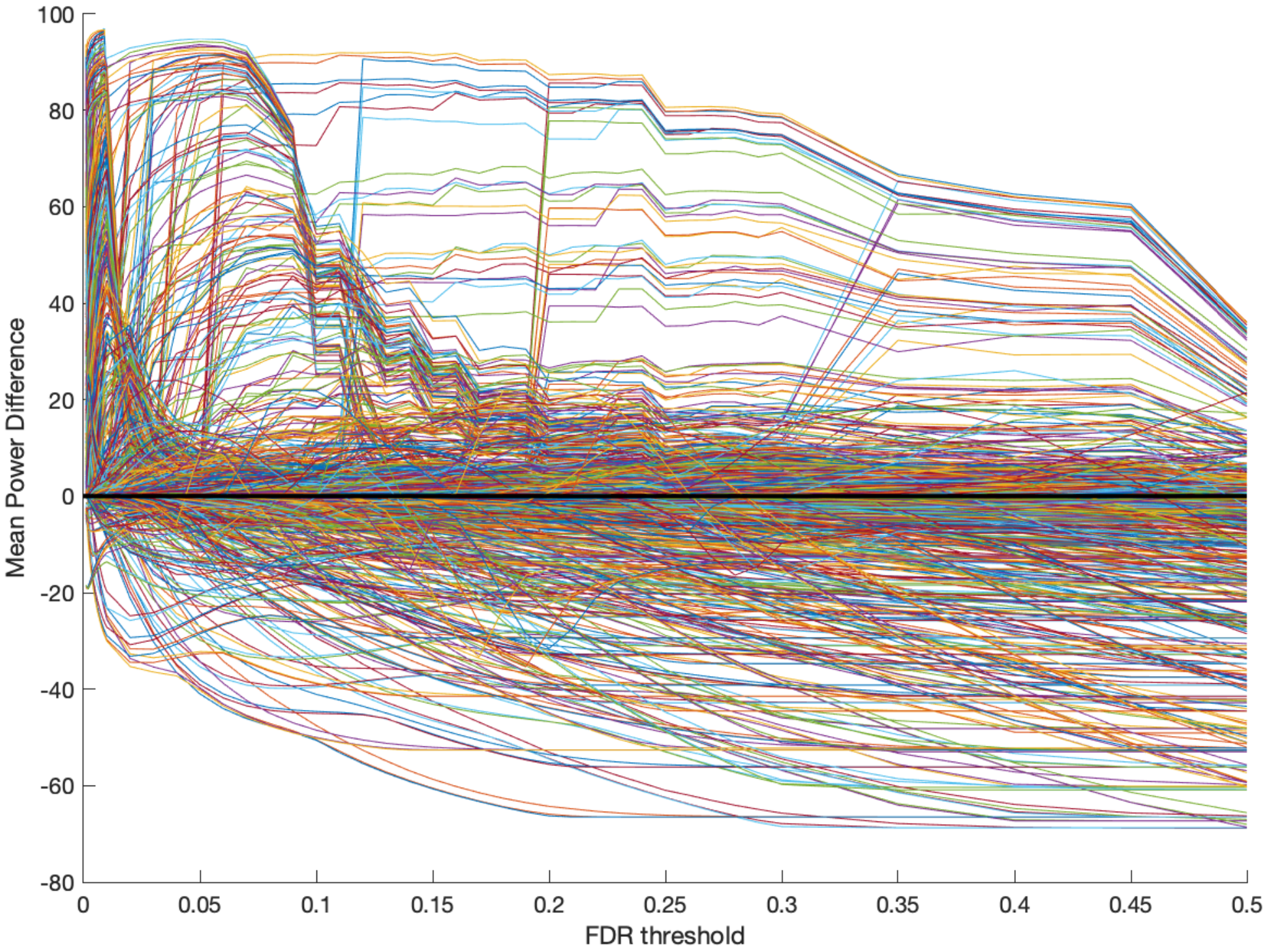} \tabularnewline
C: LF vs.~mirror & D: LF vs.~max \tabularnewline
\includegraphics[width=3in]{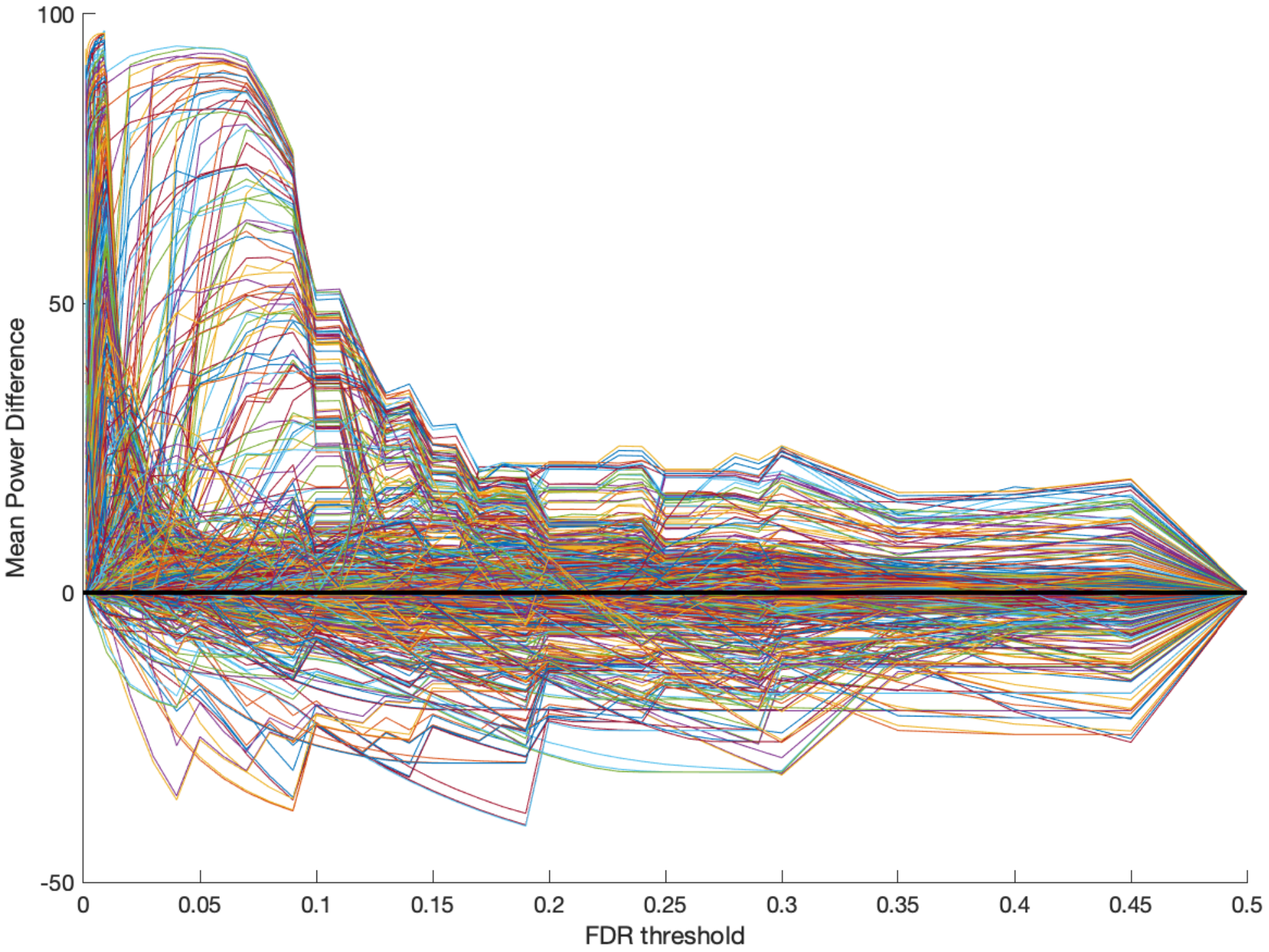}  & \includegraphics[width=3in]{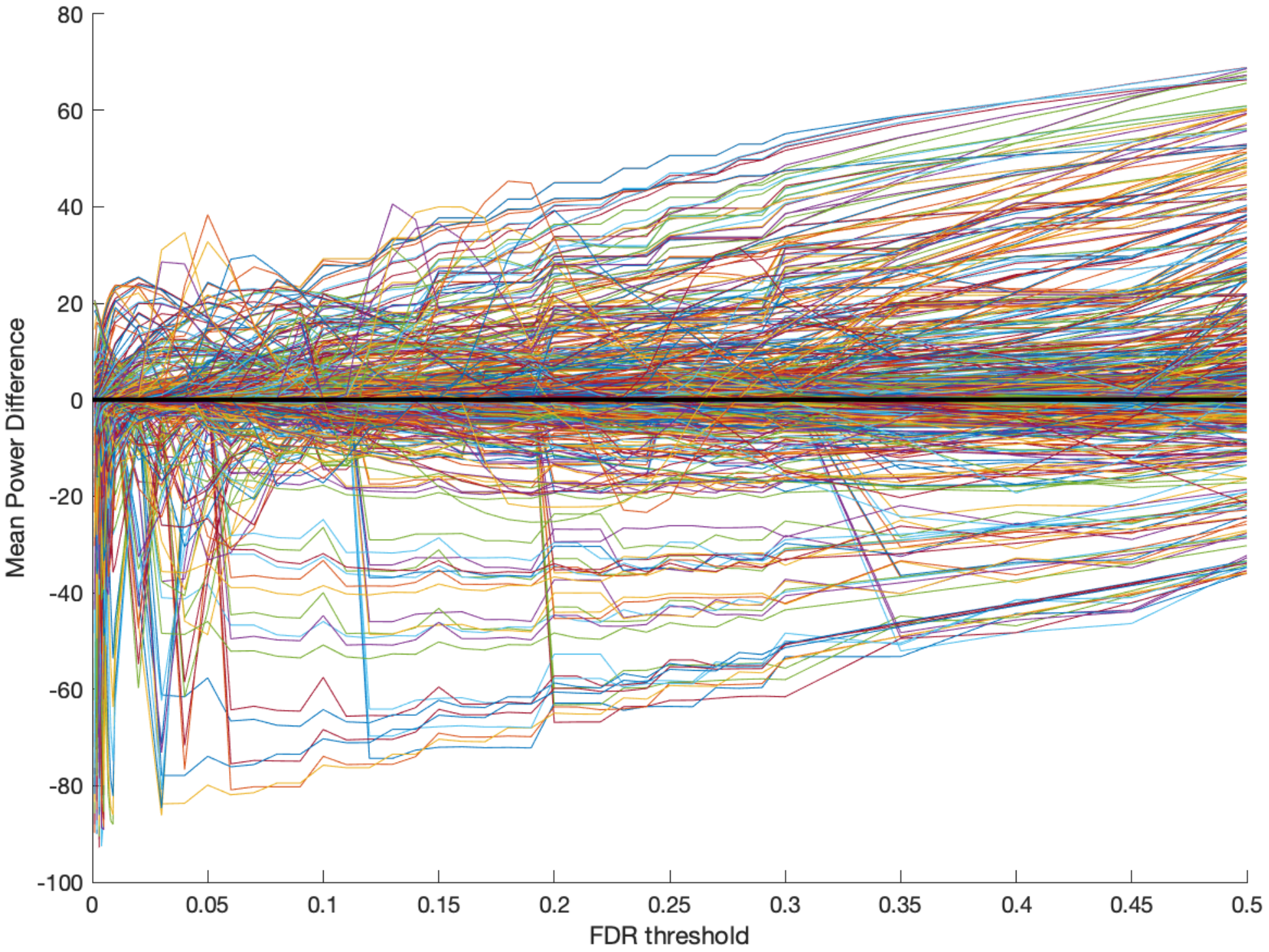} \tabularnewline
E: TDC vs.~mirror & F: TDC vs.~max \tabularnewline
\includegraphics[width=3in]{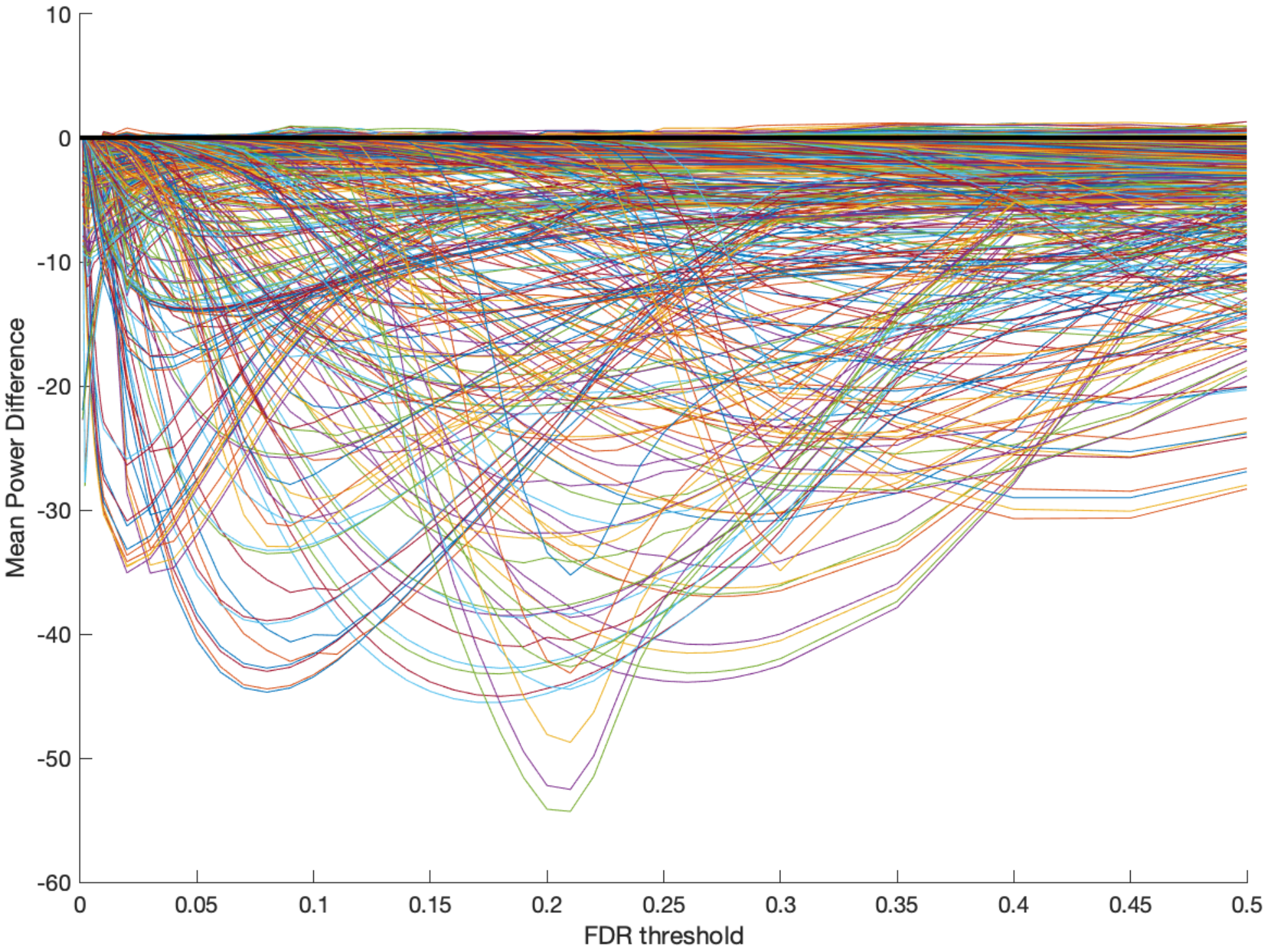} &
\includegraphics[width=3in]{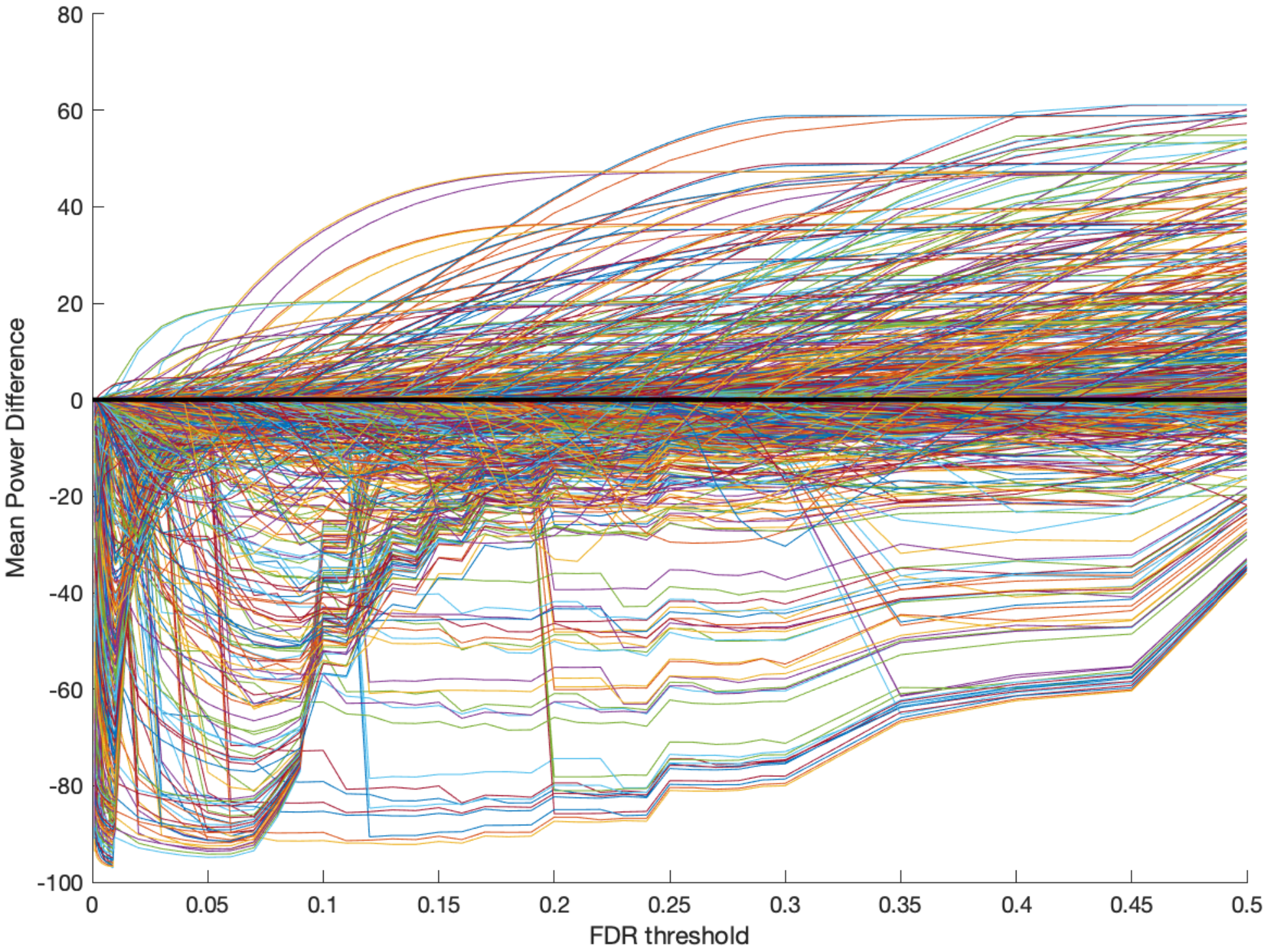} \tabularnewline
\end{tabular}\caption{\textbf{Power difference.} 
Each of the panels show the difference in the average power of the two compared methods
(positive values indicate the first method is more powerful). 
Each panel is made of 1200 curves, each of which shows the difference in power averaged
over the 10K (100K for panel A) sets. The sets were drawn simulating both calibrated and non-calibrated scores
using the experiment-specific parameter combination as described in \supsec~\ref{sec:Simulation-setup}.
The power of each method is the 10K-average (100K for panel A) percentage of false nulls
that are discovered at the given FDR threshold.
Note that figures' y-axes are on different scales.
\textbf{(A:)}  with $c=\lam=1/2$ the mirror map ($\vrp_m$) is consistently better than the randomized uniform map ($\vrp_u$); 100K draws for each of the 1200 parameter combinations. \textbf{(B-D:)}  for each of the mirror, max, and LF procedures there are numerous cases where its power
is significantly below one of the other methods. \textbf{(F:)} the mirror is consistently better than TDC. \textbf{(E:)} the max is not consistently better than TDC.
\label{fig:power_initial}}
\end{figure}

\begin{figure}
\centering %
\begin{tabular}{ll}
A: FDS  & B: FDS$_1$ \tabularnewline
\includegraphics[width=3in]{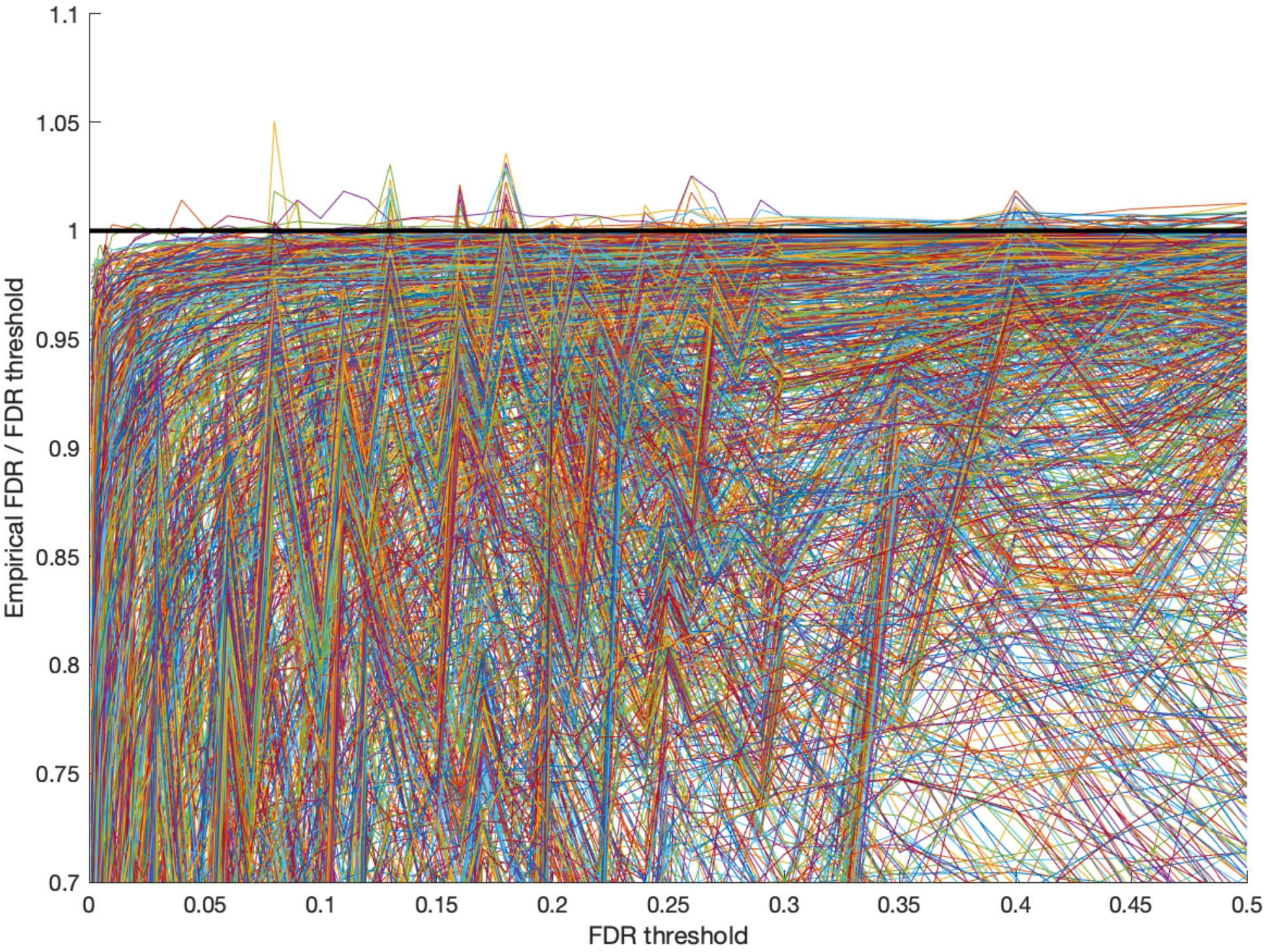}  & \includegraphics[width=3in]{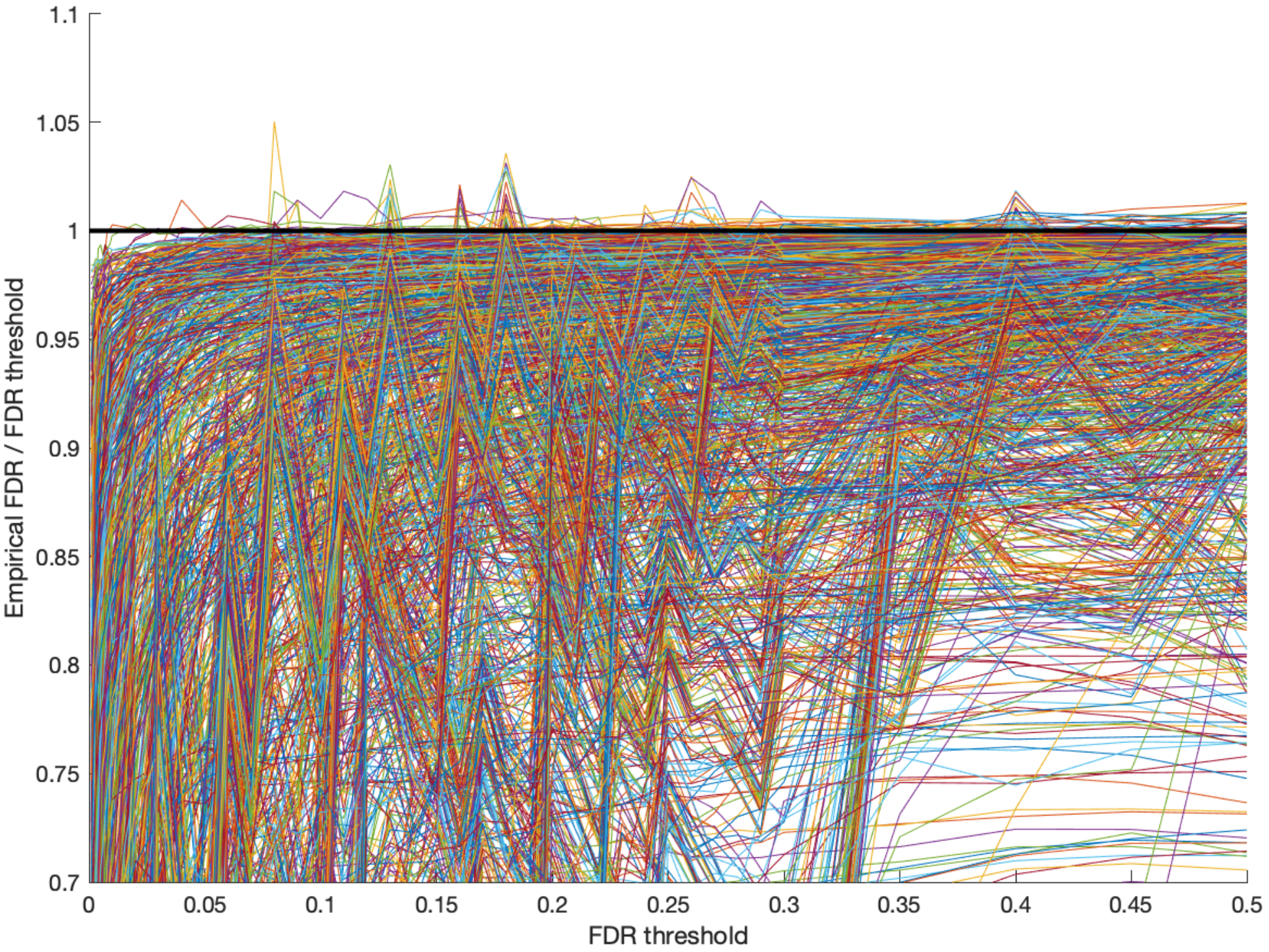} \tabularnewline
C: max & D: LBM \tabularnewline
\includegraphics[width=3in]{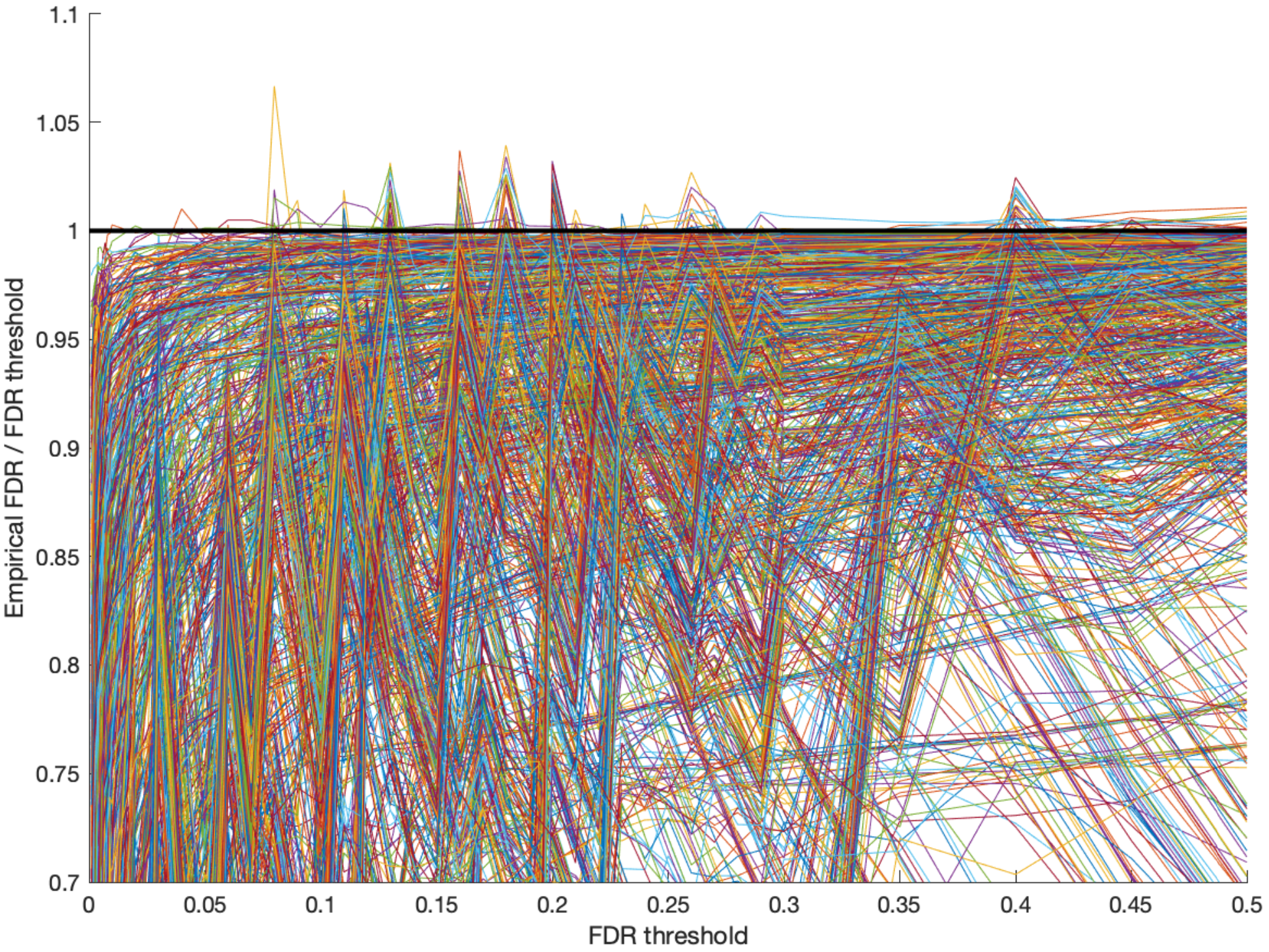}  & \includegraphics[width=3in]{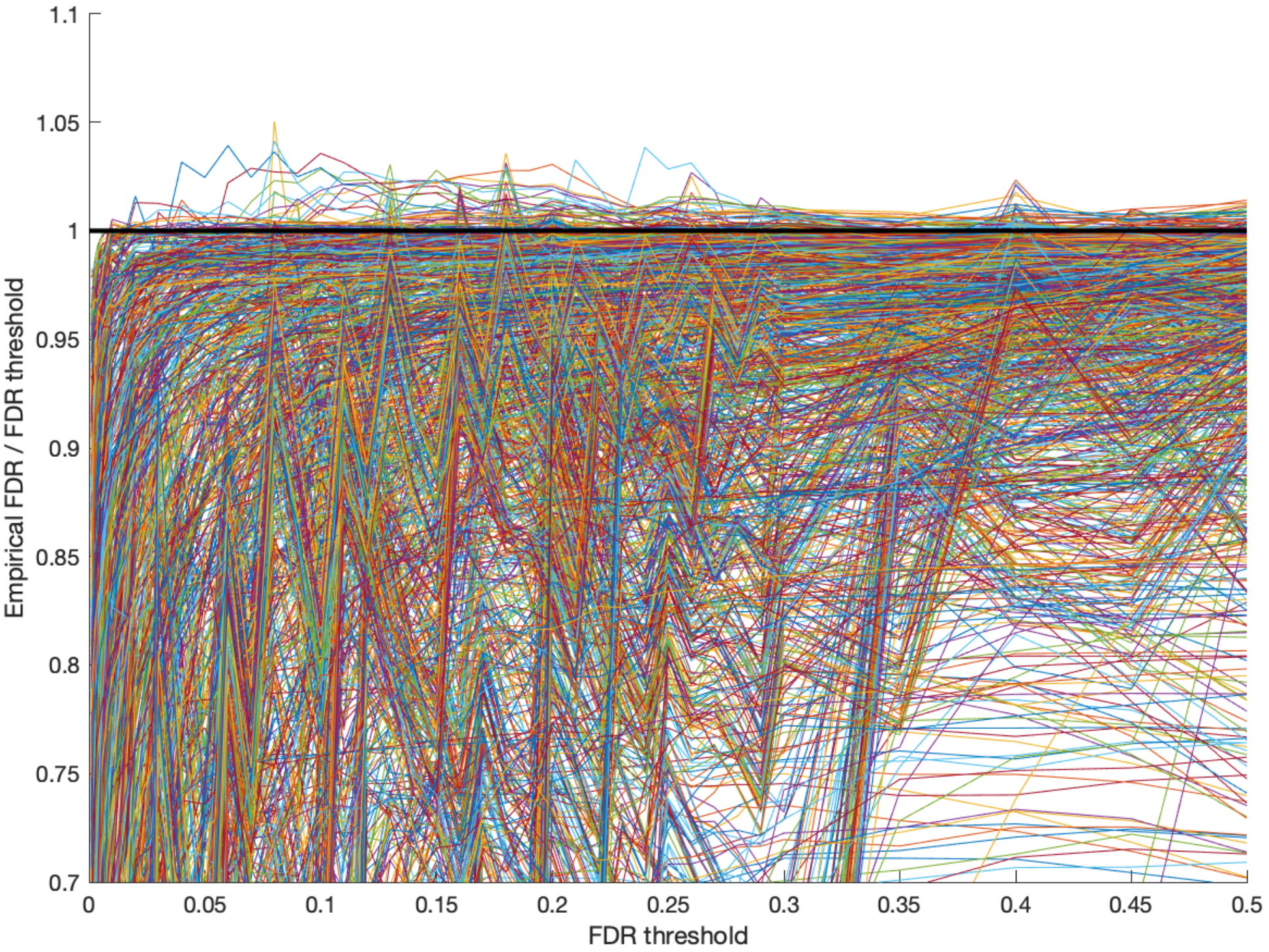} \tabularnewline
\end{tabular}\caption{\textbf{FDR control.} The panels show the ratio of the empirical FDR
to the selected FDR threshold, and each is made of 1200 curves, each
of which corresponds to one experiment involving 10K randomly drawn sets.
The empirical FDR is the 10K-sample average of the FDP of each
method's discovery list at the selected FDR threshold.
The 10K sets were drawn simulating both calibrated and non-calibrated scores
using the experiment-specific parameter combination as described in \supsec~\ref{sec:Simulation-setup}.
Notably there are relatively few cases where the empirical FDR is above the threshold (ratio $>1$),
and it is instructive to compare the ones observed in FDS, FDS$_1$ and LBM with those
we note in the max method.
Specifically, the overall maximal observed violation is 5.0\% for FDS, FDS$_1$ and LBM while it is 6.7\% for max.
Similarly, the number of curves (out of 1200) in which the maximal violation exceeds 2\% is 7 for FDS and FDS$_1$,
21 for LBM, and 24 for the max. Given that the max provably controls the FDR these simulations suggest that
FDS, FDS$_1$ and LBM essentially control the FDR as well.
\label{fig:FDR-control}}
\end{figure}

\begin{figure}
\centering %
\begin{tabular}{ll}
A: max vs.~FDS$_1$  & B: LF vs.~FDS$_1$ \tabularnewline
\includegraphics[width=3in]{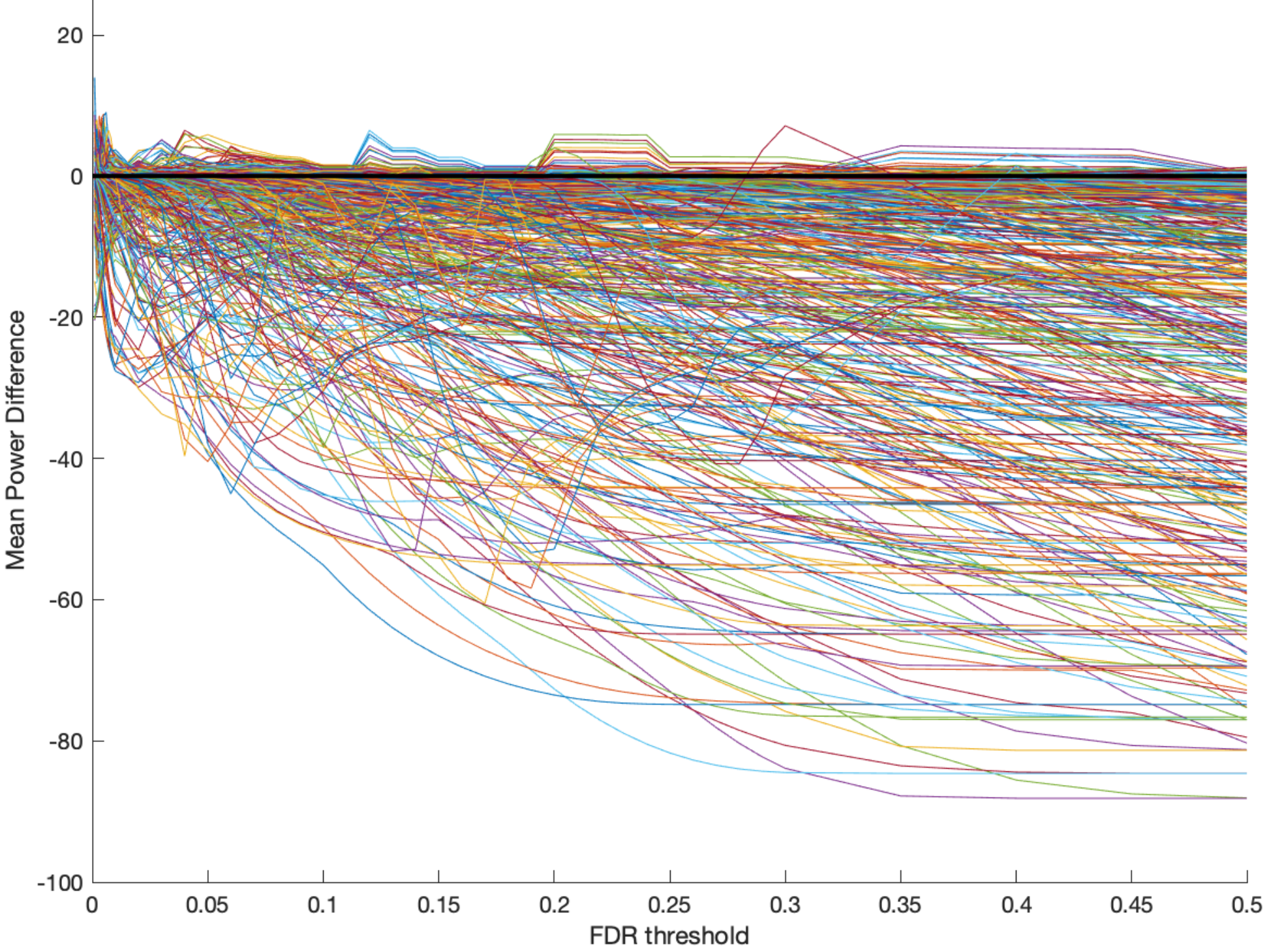}  & \includegraphics[width=3in]{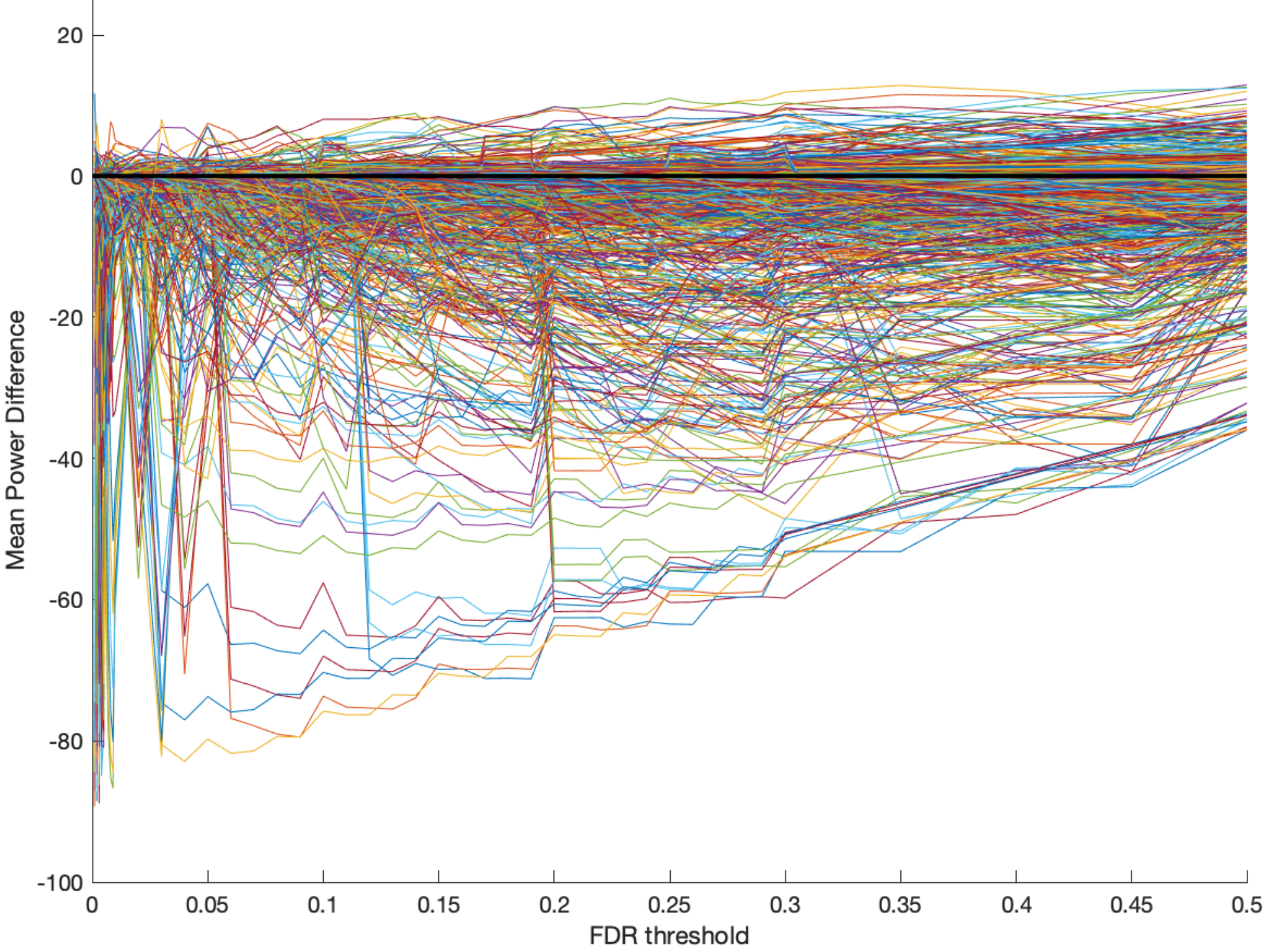} \tabularnewline
C: mirror vs.~FDS$_1$ & D: FDS vs.~FDS$_1$ \tabularnewline
\includegraphics[width=3in]{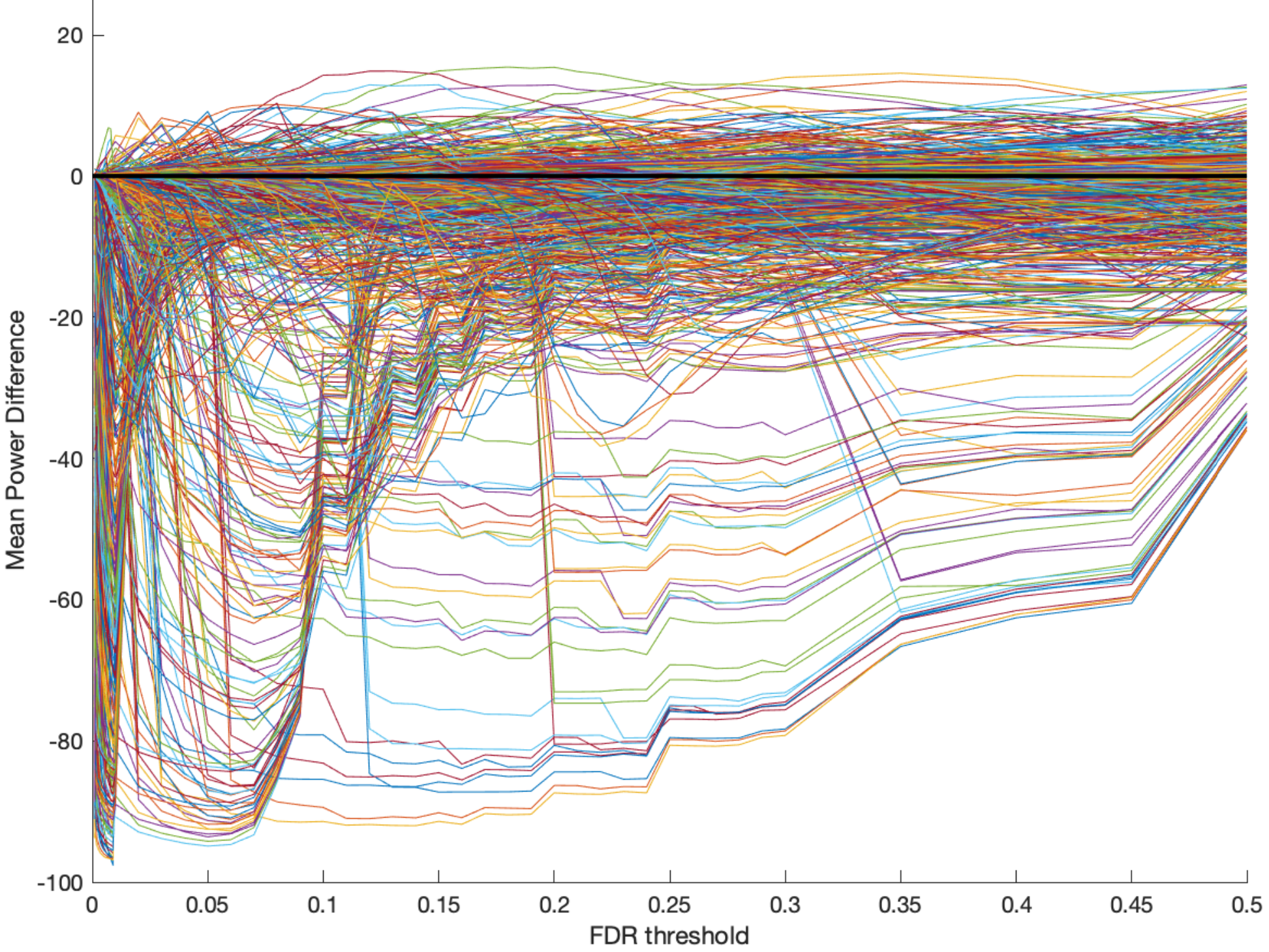}  & \includegraphics[width=3in]{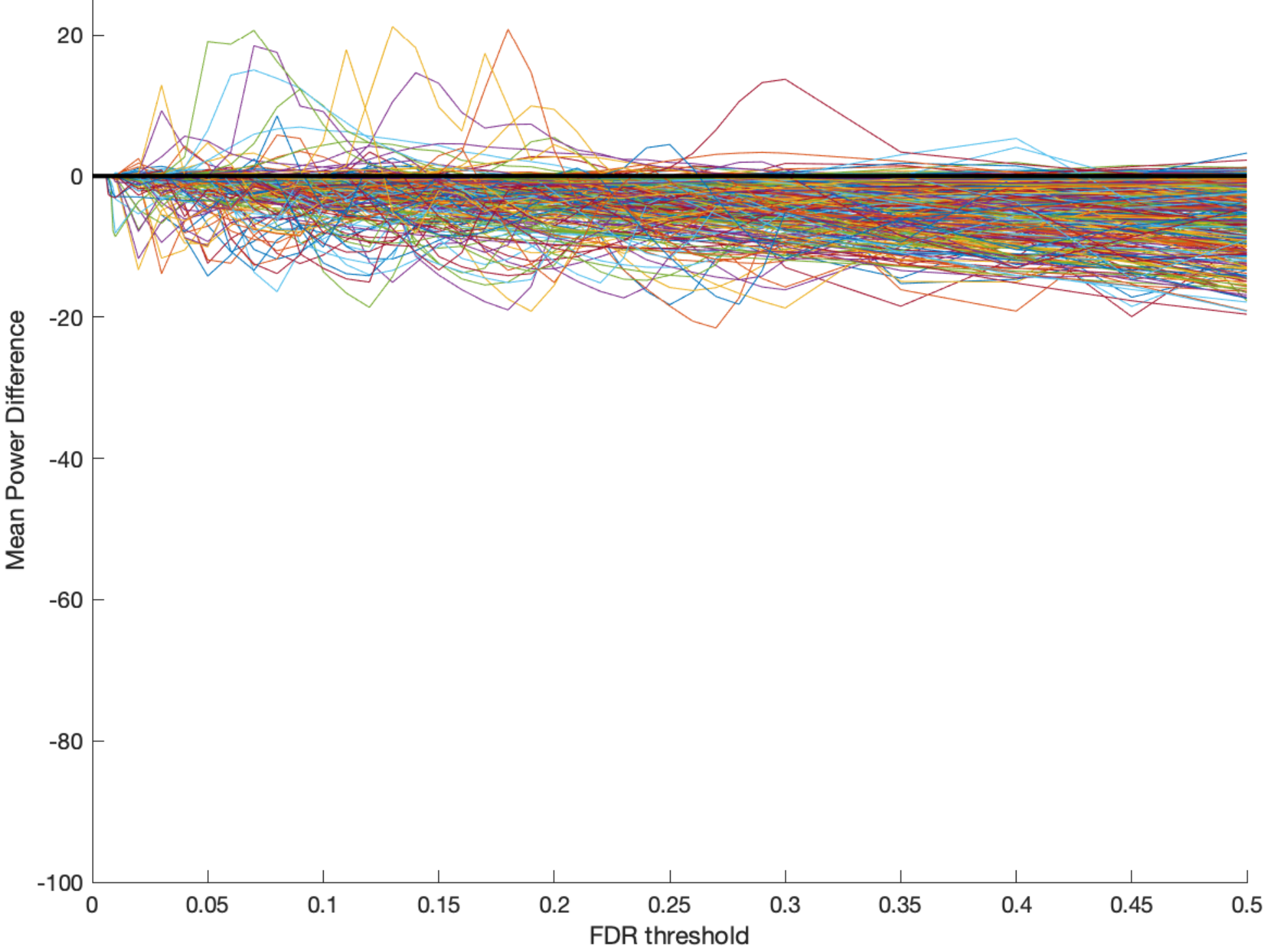} \tabularnewline
E: TDC vs.~FDS$_1$ & F: LBM vs.~FDS$_1$ \tabularnewline
\includegraphics[width=3in]{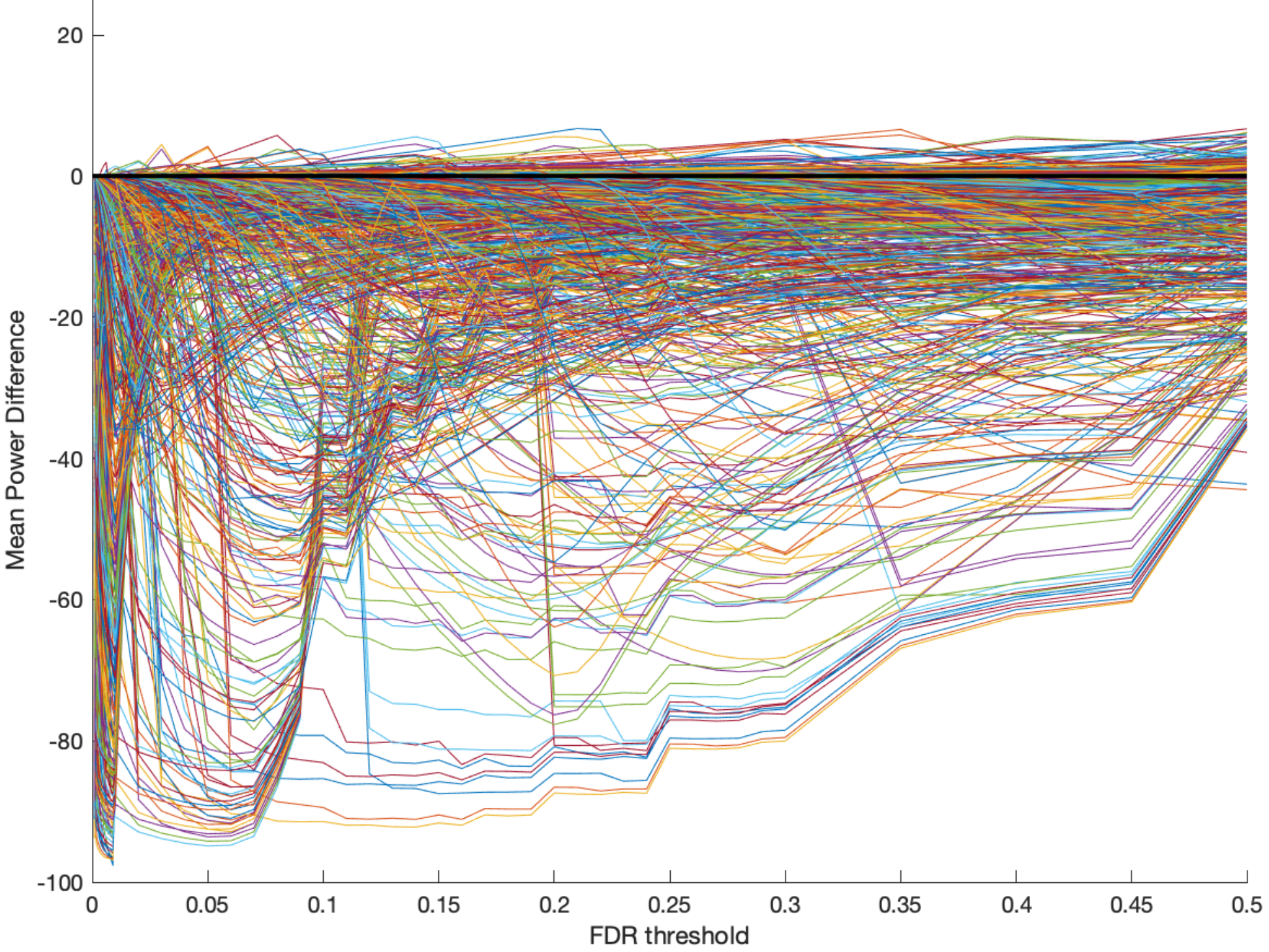} &
\includegraphics[width=3in]{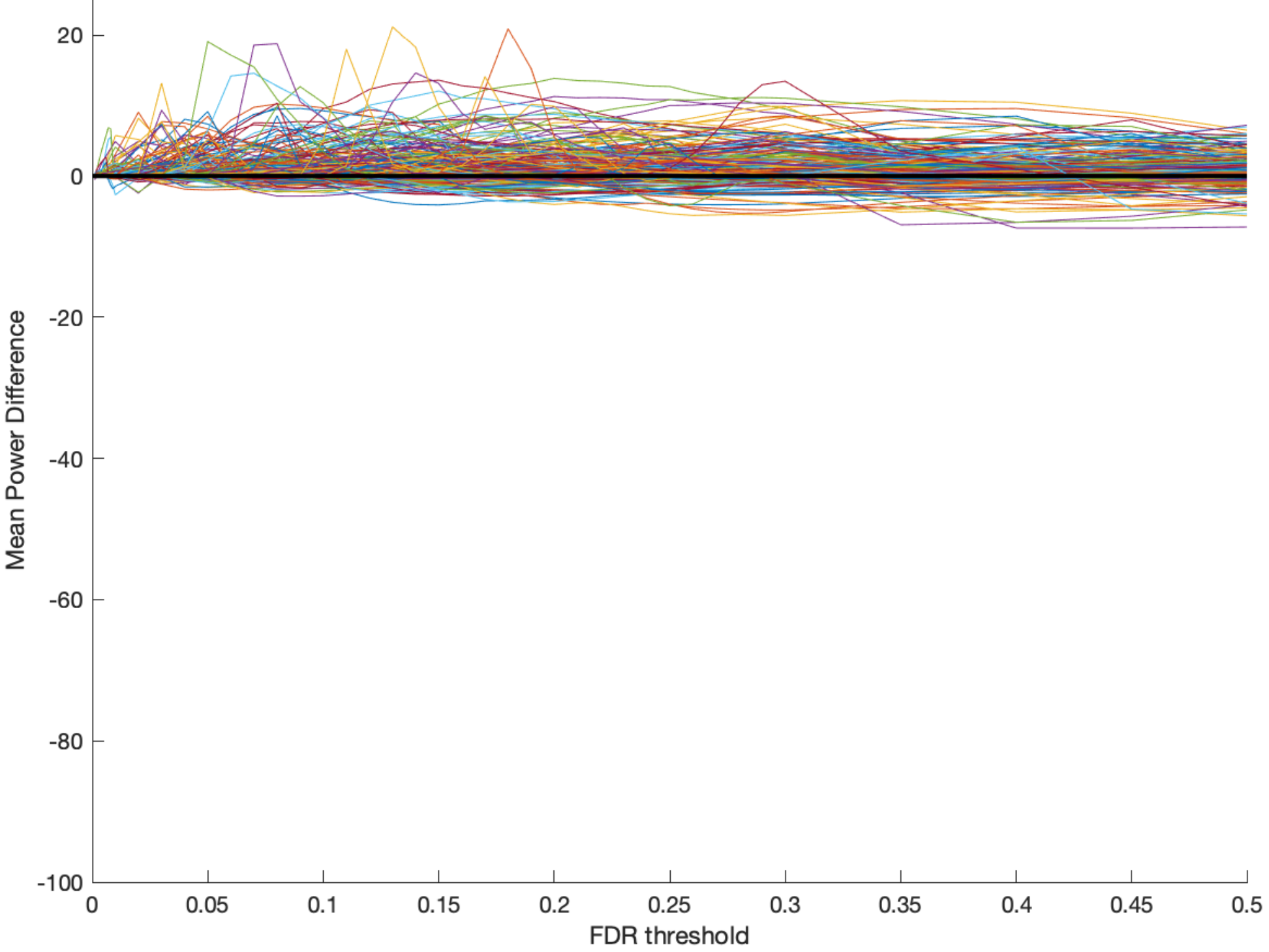} \tabularnewline
\end{tabular}\caption{\textbf{Power relative to FDS$_1$.} 
 Each of the panels show the difference between the average power of the noted method and FDS$_1$
(negative values indicate FDS$_1$ is more powerful). 
Each panel is made of 1200 curves, each of which shows the difference in power averaged
over the 10K sets. The sets were drawn simulating both calibrated and non-calibrated scores
using the experiment-specific parameter combination as described in \supsec~\ref{sec:Simulation-setup}.
The power of each method is the 10K-average percentage of false nulls
that are discovered at the given FDR threshold.
\textbf{(A-E:)} FDS$_1$ arguably offers the best compromise among the multi-decoy procedures of mirror, max, LF, and FDS, as well
as the single decoy TDC. \textbf{(F:)} LBM seems to offer an overall more powerful procedure.
\label{fig:power_FDS1}}
\end{figure}

\begin{figure}
\centering %
\begin{tabular}{ll}
A: max vs.~LBM  & B: LF vs.~LBM \tabularnewline
\includegraphics[width=3in]{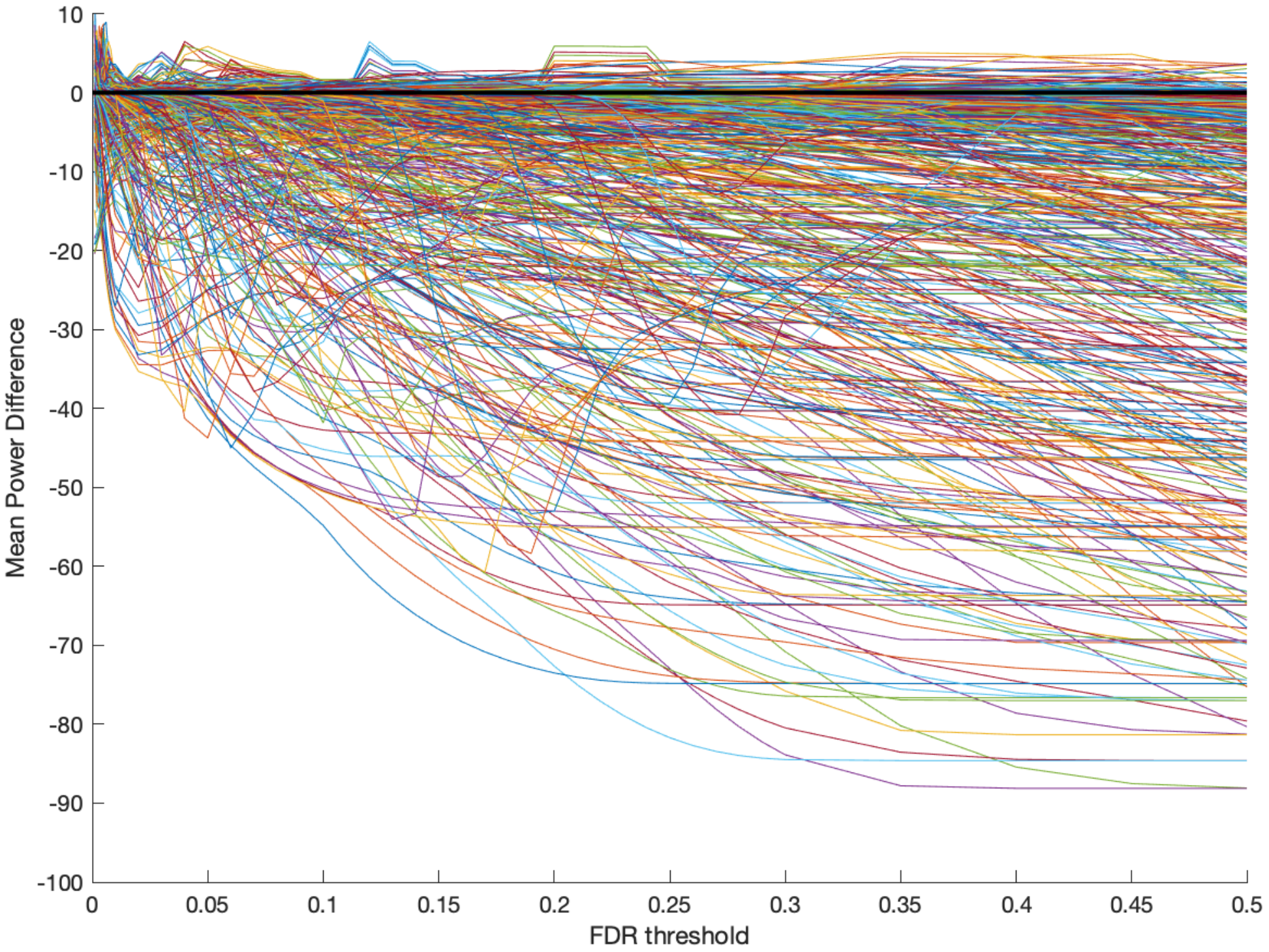}  & \includegraphics[width=3in]{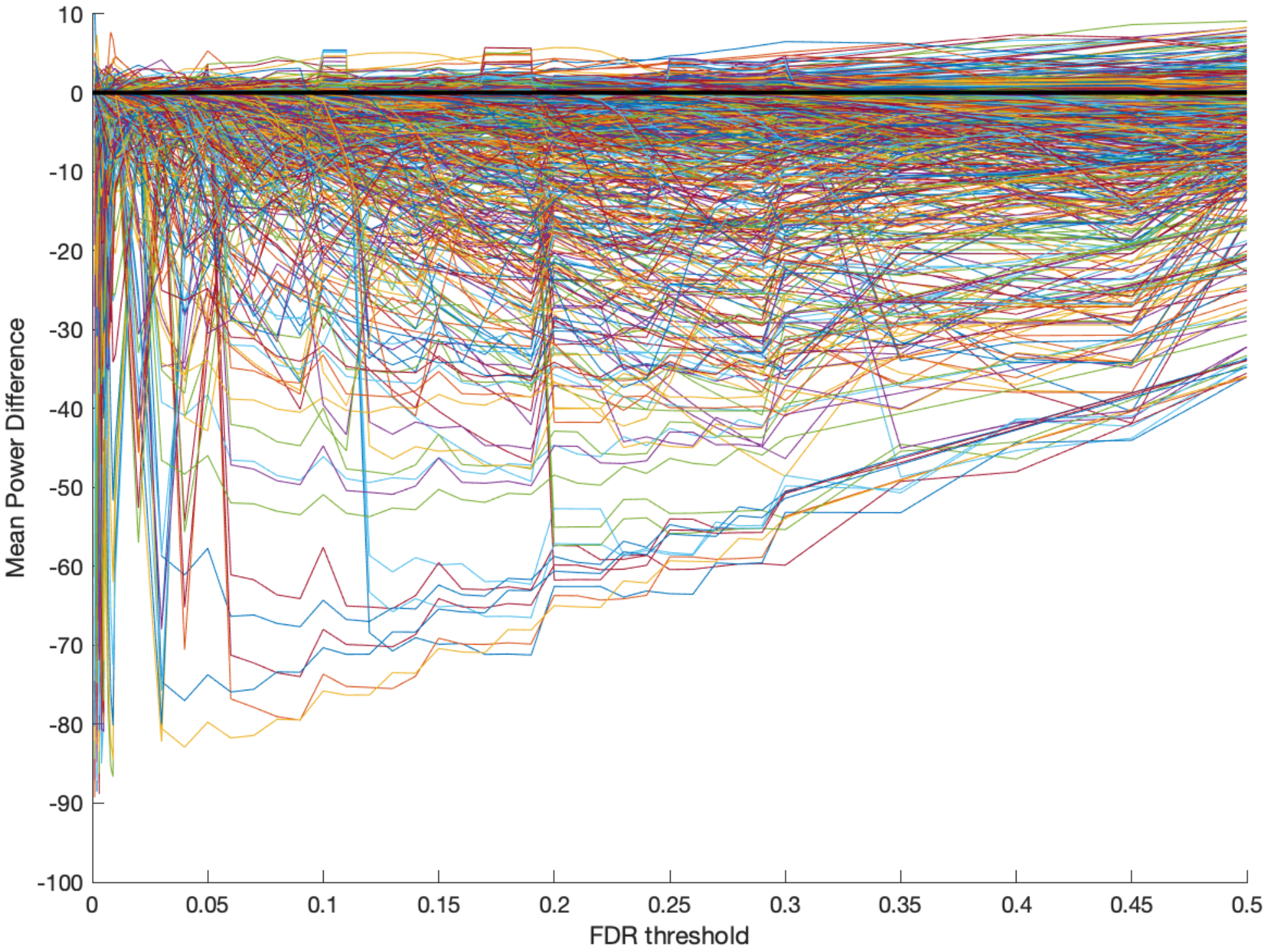} \tabularnewline
C: mirror vs.~LBM & D: FDS vs.~LBM \tabularnewline
\includegraphics[width=3in]{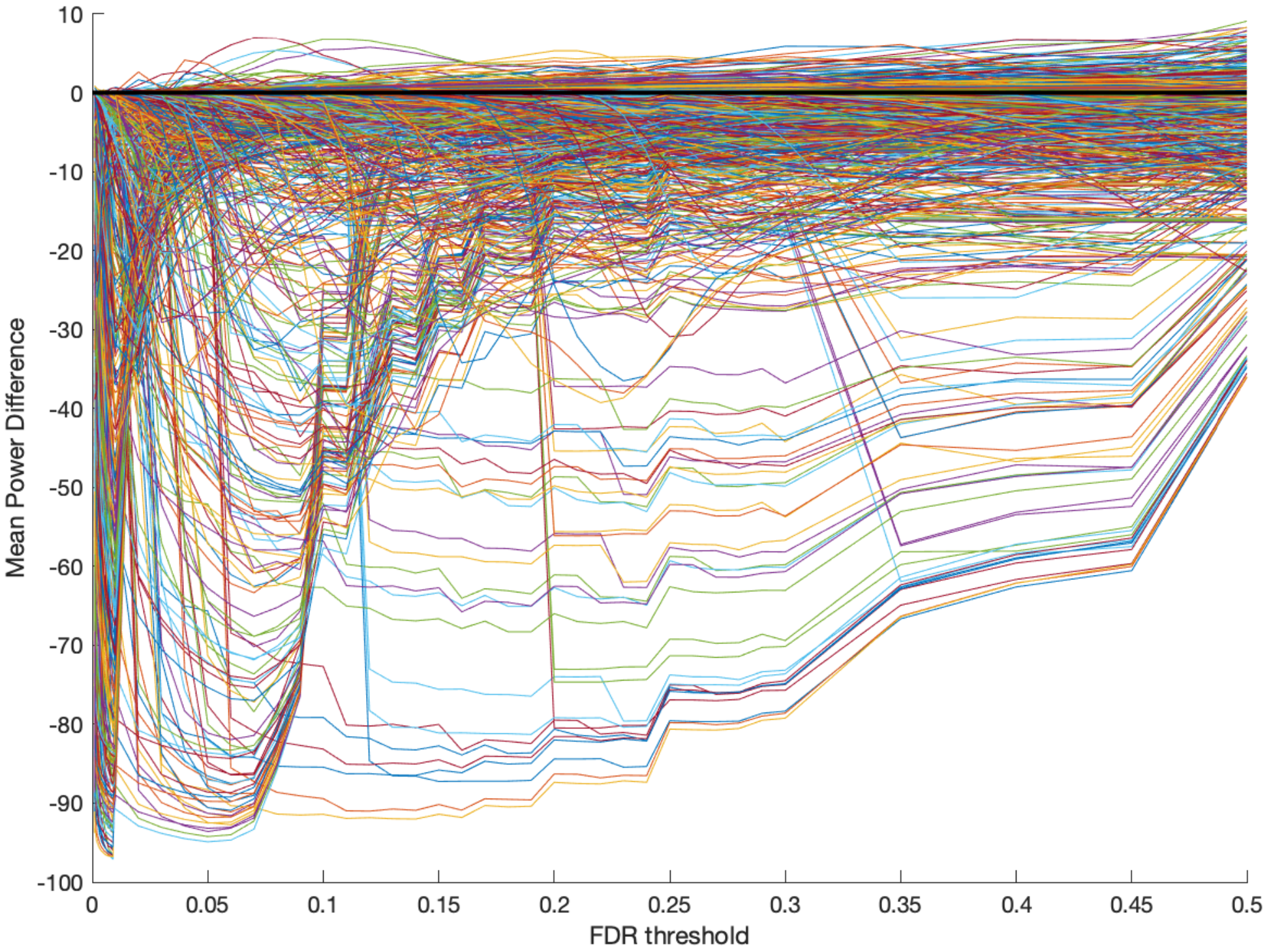}  & \includegraphics[width=3in]{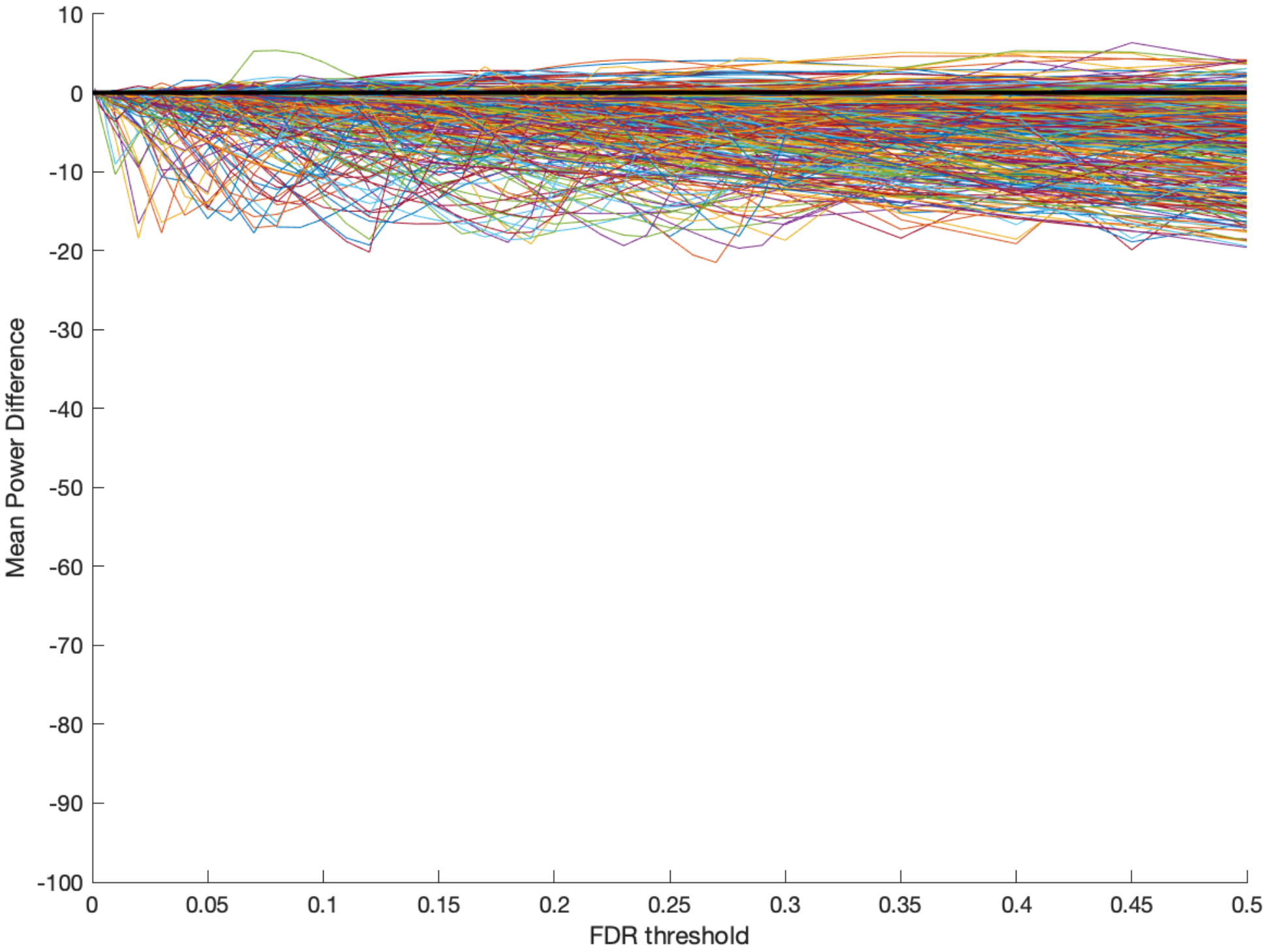} \tabularnewline
E: TDC vs.~LBM & F: aTDC vs.~LBM \tabularnewline
\includegraphics[width=3in]{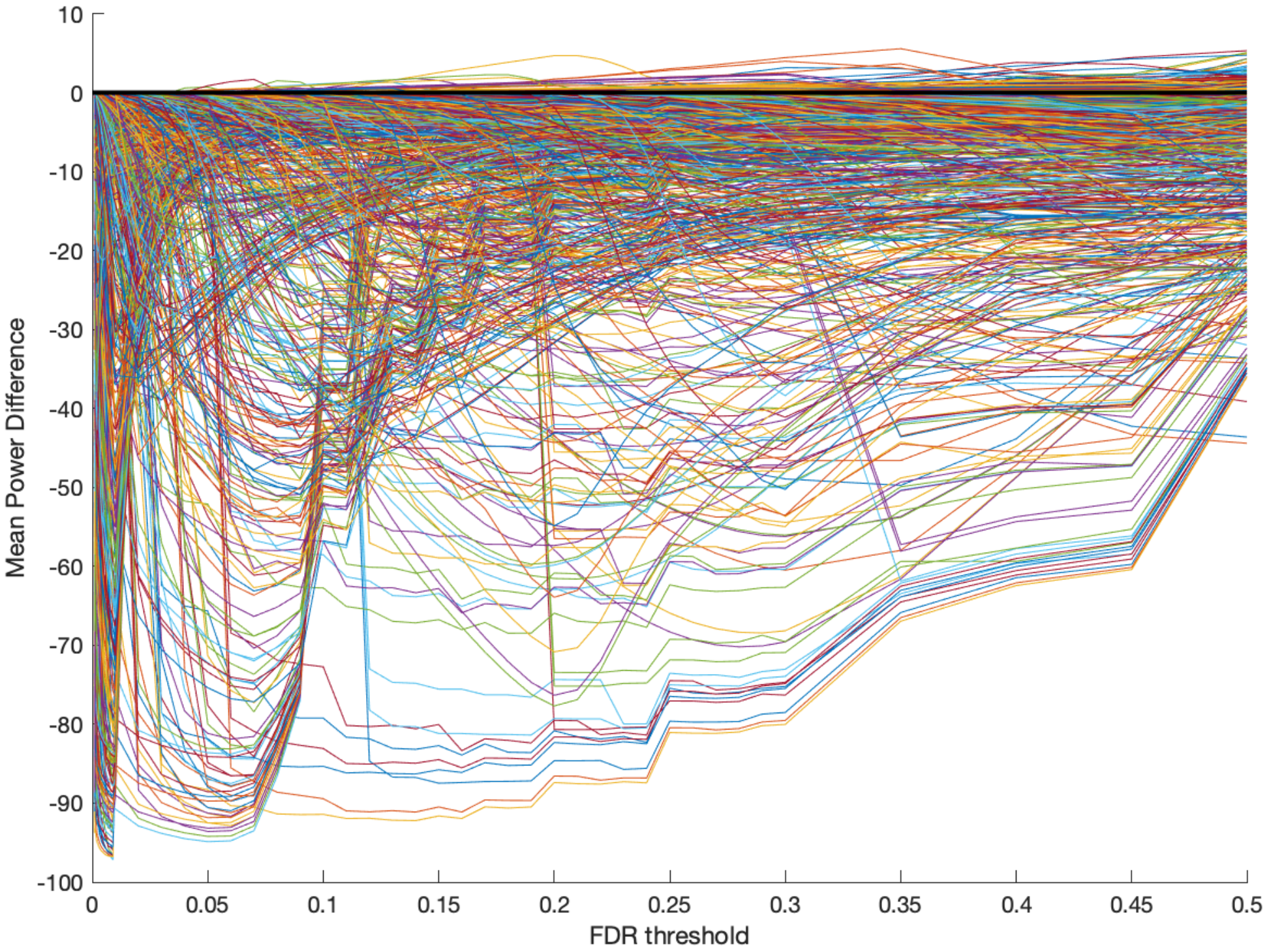} &
\includegraphics[width=3in]{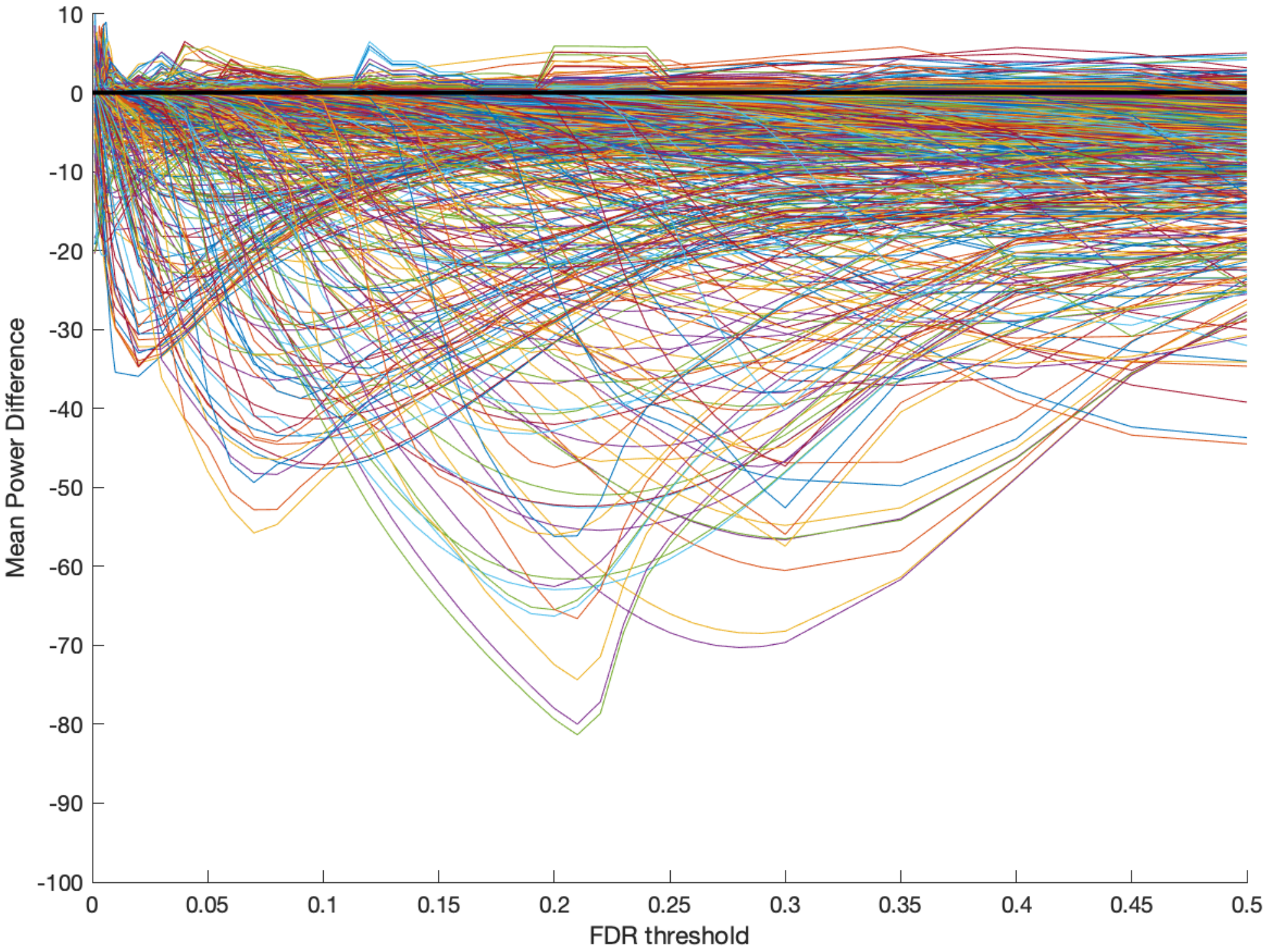} \tabularnewline
\end{tabular}\caption{\textbf{Power relative to LBM.} 
 Each of the panels shows the difference between the average power of the noted method and LBM
(negative values indicate LBM is more powerful). For LBM vs.~FDS$_1$ see panel F of \supfig~\ref{fig:power_FDS1}.
Each panel is made of 1200 curves, each of which shows the difference in power averaged
over the 10K sets. The sets were drawn simulating both calibrated and non-calibrated scores
using the experiment-specific parameter combination as described in \supsec~\ref{sec:Simulation-setup}.
The power of each method is the 10K-average percentage of false nulls
that are discovered at the given FDR threshold.
LBM seems to offer an overall optimal procedure among the competition-based procedures considered here.
\label{fig:power_LBM}}
\end{figure}

\clearpage

%
%
%
%
%
%

\end{document}